\documentclass[11pt]{article}

\usepackage[active]{srcltx}
\usepackage{hyperref}
\usepackage{amsmath,amssymb}
\usepackage[all]{xy}
\usepackage{graphicx}

 1 
\newtheorem{theorem}{Theorem}
\newtheorem{proposition}{Proposition}

\newtheorem{definition}{Definition}

\def\qed{\ifvmode\Realemovelastskip\fi
{\unskip\nobreak\hfil\penalty50\hbox{}\nobreak\hfil \hbox{\vrule
height1.2ex width1.2ex}\parfillskip=0pt \finalhyphendemerits=0
\par\smallskip}}
\def\qedr{\ifvmode\Realemovelastskip\fi
{\unskip\nobreak\hfil\penalty50\hbox{}\nobreak\hfil \hbox{
$\diamond$}\parfillskip=0pt \finalhyphendemerits=0
\par\smallskip}}

\newenvironment{proof}{\noindent{\sl Proof:~~~}}{\quad \qed}

\def\beq{\begin{equation}}
\def\eeq{\end{equation}}
\def\bea{\begin{eqnarray}}
\def\eea{\end{eqnarray}}
\def\beann{\begin{eqnarray*}}
\def\eeann{\end{eqnarray*}}
\def\beasn{\begin{sneqnarray}}
\def\eeasn{\end{sneqnarray}}
\def\ben{\begin{enumerate}}
\def\een{\end{enumerate}}
\def\bit{\begin{itemize}}
\def\eit{\end{itemize}}

\def\derpar#1#2{\displaystyle\frac{\partial{#1}}{\partial{#2}}}
\def\derpars#1#2#3{\displaystyle\frac{\partial^2{#1}}{\partial{#2}\partial{#3}}}
\def\restric#1#2{\left.#1\right|_{#2}}

\def\W{{\cal W}}
\def\C{{\cal C}}

\def\N{{\cal N}}
\def\bfxi{{\boldsymbol{\xi}}}
\def\bfalpha{{\boldsymbol{\alpha}}}
\def\bfnu{{\boldsymbol{\nu}}}
\def\bfmu{{\boldsymbol{\mu}}}
\def\bfv{{\mathbf{v}}}
\def\g{{\mathfrak{g}}}

\def\vf{{\mathfrak{X}}}
\def\df{{\mit\Omega}}
\def\Lag{{\cal L}}

\def\d{{\rm Diff}}
\def\d{{\rm d}}

\def\Nat{\mathbb{N}}

\def\R{\mathbb{R}}

\def\pr{\operatorname{pr}}
\def\Tan{{\rm T}}
\def\Lie{\mathop{\rm L}\nolimits}
\def\inn{\mathop{i}\nolimits}
\def\Cinfty{{\rm C}^\infty}

\def\codim{\operatorname{codim}}
\def\tabaddress#1{{\small\it\begin{tabular}[t]{c}#1
\\[1.2ex]\end{tabular}}}
\def\qed{\ifvmode\removelastskip\fi
{\unskip\nobreak\hfil\penalty50\hbox{}\nobreak\hfil \hbox{\vrule
height1.2ex width1.2ex}\parfillskip=0pt \finalhyphendemerits=0
\par\smallskip}}

\title{Unified formalism for higher-order variational problems and its applications in optimal control}

\author{
{\sc  Leonardo Colombo\thanks{{\bf e}-{\it mail}: leo.colombo@icmat.es} }\\
\vspace{5mm}
   \tabaddress{Instituto de Ciencias Matem\'{a}ticas (CSIC-UAM-UC3M-UCM). \\
   C/ Nicol\'{a}s Cabrera 15. 28049 Madrid. Spain} \\
{\sc  Pedro Daniel Prieto-Mart\'\i nez\thanks{{\bf e}-{\it mail}: peredaniel@ma4.upc.edu} }\\
   \tabaddress{Departamento de Matem\'atica Aplicada IV.
   Edificio C-3, Campus Norte UPC\\
   C/ Jordi Girona 1. 08034 Barcelona. Spain}
}

\begin{document}

\maketitle

\thispagestyle{empty}

\begin{abstract}
In this paper we consider an intrinsic point of view to describe
the equations of motion for higher-order variational problems
with constraints
on higher-order trivial principal bundles. Our techniques are an
adaptation of the classical Skinner-Rusk approach for the case
of Lagrangian dynamics with higher-order constraints. We study a
regular case where it is possible to establish a symplectic
framework and, as a consequence, to obtain a unique vector field
determining the dynamics. As an interesting application we deduce
the equations of motion for optimal control of underactuated
mechanical systems defined on principal bundles.
\end{abstract}


\vbox{\raggedleft AMS s.\,c.\,(2000): 70H50, 22E70, 49J15, 53C80.}\null
\markright{\rm L. Colombo, P.D. Prieto-Mart\'{\i}nez:
\sl Higher-order mechanical systems \& optimal control appl.}

\clearpage

\tableofcontents

\clearpage

\section{Introduction}
\label{sec:Introduction}

The aim of this work is to describe, in an intrinsic way,
higher-order Euler-Lagrange equations on trivial principal bundles.
The main motivation is the analysis of a class of optimal control
problems that are important in a wide range of contexts, such as the
search of less cost devices in industrial processes, aerospace and
defense, marine and automotive systems, electro-mechanical systems
and robotic, among others. The optimal control problem consists in
finding a trajectory of the state variables and control inputs,
solution of the controlled Euler-Lagrange equations given initial
and final conditions, and minimizing a cost function.

The study of higher-order tangent bundles has been developed in the last decades for different
reasons, as a training field to understand field theory,
theoretical physics, relativistic mechanics, classification of
higher-order symmetries, among others
\cite{proc:Cantrijn_Crampin_Sarlet86,art:Carinena_Lopez92,
art:DeLeon_Lacomba89,art:DeLeon_Martin94_1,art:DeLeon_Pitanga_Rodrigues94,
book:DeLeon_Rodrigues85,art:Gracia_Pons_Roman91,art:Krupkova00,book:Saunders89}.
In the last decade,
higher-order variational problems had an extraordinary impact
in the design and planning of trajectories, interpolation problems
in Riemannian manifolds, optimization problems, optimal control
applications, higher-order jet groups and particle methods, image
registration methods for computational anatomy, etc.
\cite{art:Camarinha_SilvaLeita_Crouch01,art:Colombo_Martin_Zuccalli10,
art:GayBalmaz_Holm_Meier_Ratiu_Vialard12_1,art:GayBalmaz_Holm_Meier_Ratiu_Vialard12_2,
art:Grillo_Zuccalli12,art:Hinkle_Muralidharan_Fletcher12,art:Hussein_Bloch07,
art:Jacobs_Ratiu_Desbrun12}

In 1983, R. Skinner and R. Rusk introduced a formulation for the dynamics
of an autonomous mechanical system which combines the Lagrangian and
Hamiltonian features \cite{art:Skinner_Rusk83}.
The aim of this formulation is to obtain a
common framework for both regular and singular dynamics, obtaining
simultaneously the Lagrangian and Hamiltonian formulations of the
dynamics. Over the years, the Skinner-Rusk framework, or unified framework, has been extended
in many directions: explicit time-dependent systems using a jet bundle
language \cite{art:Barbero_Echeverria_Martin_Munoz_Roman08,art:Cortes_Martinez_Cantrijn02},
vakonomic mechanics and the comparison between the
solutions of vakonomic and nonholonomic mechanics
\cite{art:Cortes_deLeon_Martin_Martinez02},
higher-order dynamical systems \cite{art:Prieto_Roman11,art:Prieto_Roman12_1},
first-order and higher-order classical field theories
\cite{art:Campos_DeLeon_Martin_Vankerschaver09,art:Echeverria_Lopez_Marin_Munoz_Roman04,art:Vitagliano10},
and optimal control applications
\cite{art:Barbero_Echeverria_Martin_Munoz_Roman07,art:Benito_Martin05,art:Colombo_Martin11}.

There is an interesting class of mechanical control systems,
underactuated mechanical systems, which are characterized by the
fact that there are more degrees of freedom than actuators. This
type of systems is quite different from a mathematical and
engineering perspective than fully actuated control systems (that
is, where all the degrees of freedom are actuated). Underactuated
systems include spacecraft, underwater vehicles, mobile robots,
helicopters, wheeled vehicles, mobile robots, underactuated
manipulators... (see
\cite{art:Baillieul99,book:Bloch03,book:Bullo_Lewis05} and
references therein).

Our main objective in this paper is to characterize geometrically
the equations of motion for a higher-order autonomous system with
constraints using an extension of the Skinner-Rusk formalism for
higher-order trivial principal bundles, and apply this to the
optimal control problems of an underactuated mechanical system. The
main results of this work can be found in Section 3, where we give a
general method to deal with explicit and implicit, constrained and
unconstrained mechanical systems.

The organization of the paper is as follows.
In Section \ref{sec:Background} we
introduce some geometric constructions which are used along the
paper. In particular, the Skinner-Rusk formalism for first-order
mechanical systems, the Gotay-Nester-Hinds algorithm, some geometric aspects
of higher-order tangent bundles, and Hamilton equations for a
Hamiltonian system defined on the cotangent bundle of a higher-order
trivial principal bundle. In Section \ref{sec:HOVariationalProblems}
we introduce the Pontryagin
bundle $\Tan^{*}(\Tan^{k-1}M) \times G\times k\g\times k\g^*$,
where we introduce the dynamics using a
presymplectic Hamiltonian formalism, and we deduce the $k$th-order
Euler-Lagrange equations in this context. Since the system is
presymplectic, it is necessary to analyze the consistency of the
dynamics using a constraint algorithm. We show that our techniques
are easily adapted to the case of constrained dynamics. As an
illustration of the applicability of our formulation, we analyze in
Section \ref{sec:OptimalControlUnderactuated}
the case of underactuated control of mechanical systems
and, as a particular example, the optimal control problem of an
underactuated vehicle in $SE(2)\times\mathbb{S}^{1}.$

All the manifolds are real, second countable and $\Cinfty$.
The maps and the structures are assumed to be $\Cinfty$.
Sum over repeated indices is understood.

\section{Mathematical background}
\label{sec:Background}

In this section we give the notation used along this work, and
the basic mathematical background needed about higher-order tangent
bundles. There is also a sketch of the Gotay-Nester-Hinds algorithm and
the Skinner-Rusk formalism for first-order systems. Finally we will derive
Hamilton equations for higher-order trivial principal bundles.

\subsection{The Lagrangian-Hamiltonian unified formalism}
\label{sec:SkinnerRusk}

(See \cite{art:Skinner_Rusk83} for details.)

Let $Q$ be a $n$-dimensional smooth manifold modeling the
configuration space of a first-order dynamical system with $n$
degrees of freedom, and $\Lag \in \Cinfty(\Tan Q)$ a Lagrangian
function describing the dynamics of the system.

Let us consider the bundle
$$
\W = \Tan Q \times_{Q} \Tan^{*}Q \, .
$$
This bundle is endowed with canonical projections over each factor,
namely $\pr_1 \colon \W \to \Tan Q$ and $\pr_2 \colon \W \to
\Tan^{*}Q$. Using these projections, and the canonical projections
of the tangent and cotangent bundle of $Q$, we introduce the
following diagram
$$
\xymatrix{
\ & \ & \Tan Q \times_Q \Tan^*Q \ar[dll]_{\pr_1}  \ar[drr]^{\pr_2} & \ & \ \\
\Tan Q \ar[drr]_{\tau_Q} & \ & \ & \ & \Tan^*Q \ar[dll]^{\pi_Q} \\
\ & \ & Q & \ & \
}
$$

Local coordinates in $\W$ are constructed as follows:
if $(U,\varphi)$ is a local chart of $Q$ with $\varphi=(q^{A})$,
$1 \leqslant A \leqslant n$, then the induced local charts in $\Tan Q$ and
$\Tan^{*}Q$ are $(\tau_{Q}^{-1}(U),(q^A,v^A))$
and $(\pi_{Q}^{-1}(U),(q^A,p_A))$, respectively.
Therefore, the natural coordinates in $\W$ are $(q^{A},v^{A},p_A)$.
Observe that $\dim\W = 3n$. Using these coordinates, the above
projections have the following local expressions
$$
\pr_1(q^A,v^A,p_A) = (q^A,v^A)\quad ; \quad
\pr_2(q^A,v^A,p_A) = (q^A,p_A) \, .
$$

The bundle $\W$ is endowed with some canonical geometric structures.
First, let $\theta\in\df^1(\Tan^{*}Q)$ be the Liouville $1$-form on
the cotangent bundle and $\omega = - \d\theta \in
\df^{2}(\Tan^{*}Q)$ the canonical symplectic form on $\Tan^{*}Q$.
From this we can define a $2$-form $\Omega$ in $\W$ as
$$
\Omega := \pr_2^*\omega \in \Omega^{2}(\W) \, .
$$
It is clear that $\Omega$ is a closed 2-form, since
$$
\Omega = \pr_{2}^{*}(-\d\theta) = -\d\pr_{2}^{*}\theta \, .
$$
Nevertheless, this form is degenerate, and therefore is a
presymplectic form. This is easy to check in coordinates. Bearing in mind
the local expression of the canonical symplectic form of the cotangent bundle,
which is $\omega = \d q^{A} \wedge \d p_A$, and the local expression
of the projection $\pr_2$ given above, we have
$$
\Omega =  \pr_{2}^{*}(\d q^{A} \wedge \d
p_A) = \pr_{2}^{*}(\d q^{A}) \wedge \pr_{2}^{*}(\d p_A) =
\d\pr_{2}^{*}(q^{A}) \wedge \d\pr_{2}^{*}(p_A) = \d q^{A} \wedge \d
p_A \, .
$$
From this local expression, it is clear that
the kernel of $\Omega$ is given locally by,
$$
\ker\Omega = \left\langle \derpar{}{v^A} \right\rangle = \vf^{V(\pr_2)}(\W) \, ,
$$
where $\vf^{V(\pr_2)}(\W)$ denotes the module of vector fields of $\W$
which are vertical with respect to the projection $\pr_2$
(that is, $\vf^{V(\pr_2)}(\W) = \ker(\Tan\pr_2)$).
Therefore the $2$-form $\Omega$ is degenerate.

\begin{definition}
Let $p \in Q$ be a point, $v_p \in \Tan_pQ$ a tangent vector at $p$,
and $\alpha_p \in \Tan^*_pQ$ a covector on $p$. Then we define the
\textnormal{coupling function} $\C \in \Cinfty(\W)$ as
$$
\begin{array}{rcl}
\C \colon \W & \longrightarrow & \R \\
(p,v_p,\alpha_p) & \longmapsto & \langle\alpha_p,v_p\rangle
\end{array}
$$
where $\langle \alpha_p, v_p \rangle \equiv \alpha_p(v_p)$ is the canonical pairing between
elements of $\Tan_pQ$ and $\Tan^{*}_pQ$.
\end{definition}

If we consider a local chart on $p \in Q$ such
that $\alpha_p = \restric{p_A\d q^{A}}{p}$, $v_p = v^{A}
\restric{\derpar{}{q^A}}{p}$, then the local expression of $\C$ is
$$
\C(p,v_p,\alpha_p) = \langle \alpha_p,v_p \rangle
= \left\langle  \restric{p_A \d q^{A}}{p},\restric{v^{A}\derpar{}{q^A}}{p} \right\rangle = \restric{p_Av^{A}}{p} \, .
$$

Finally, we define the Hamiltonian function $H\in\Cinfty(\W)$ by
$$
H = \C - \pr_{1}^{*}\Lag\, ,
$$
whose local expression is
$$
H(q^A,v^A,p_A) = p_Av^A - \Lag(q^A,v^A) \, .
$$

Hence, we have constructed a presymplectic Hamiltonian system
$(\W,\Omega,H)$. The dynamics for this systems is given by equation
$$
\inn(X)\Omega = \d H \, ,
$$
where $X \in \vf(\W)$ is the Hamiltonian vector field of the system.

\subsection{The constraint algorithm}
\label{sec:GNHAlgorithm}

In this subsection we briefly review the constraint algorithm
for presymplectic systems.
(See \cite{art:Gotay_Nester79,art:Gotay_Nester80,art:Gotay_Nester_Hinds78} for
details).

By definition, if $(M_1,\Omega)$ is a symplectic manifold then the
equation
\begin{equation}\label{pres}
\inn(X)\Omega = \alpha
\end{equation}
has a unique solution $X\in\mathfrak{X}(M_1)$ for every $\alpha \in
\df^{1}(M_1)$ that we consider. Nevertheless, if $\Omega$ is closed
and degenerate (that is, presymplectic), then the above equation may
not have a solution defined on the whole manifold $M_1$, but only in
some points of $M_1$. The tuple $(M_1,\Omega,\alpha)$ is said to be a
presymplectic system. The aim of the \textsl{Gotay-Nester-Hinds
algorithm}, or \textsl{constraint algorithm}, is to find a final
submanifold $M_f \hookrightarrow M_1$ such that the equation
\eqref{pres} has solutions in $M_f$ (if such submanifold exists).
More precisely, the constraint algorithm returns the maximal
submanifold $M_f$ of $M_1$ such that there exists a vector field $X \in \vf(M_f)$
satisfying equation (\ref{pres}) with support on
$M_f.$

The algorithm proceeds as follows. Since $\Omega$ is degenerate,
then equation (\ref{pres}) has no solution in general, or the
solutions are not defined everywhere. In the most favorable case,
equation (\ref{pres}) admits a global (but not necessarily unique)
solution $X \in \vf(M_1)$. Otherwise, we select the subset of points
of $M_1$, where such a solution exists, that is,
\begin{align*}
M_2 &:= \{ p \in M_1 \colon \mbox{there exists } X_p \in \Tan_pM_1 \mbox{ satisfying } \inn(X_p)\Omega_p = \alpha_p \} \\
&\ = \{ p \in M_1 \colon (\inn(Y)\alpha)(p) = 0 \mbox{ for every } Y \in \ker\Omega \} \, ,
\end{align*}
and we assume that it is a submanifold of $M_1$. Then, equation
\eqref{pres} admits a solution $X$ defined at all points of $M_2$,
but $X$ is not necessarily tangent to $M_2$, and thus it does not
necessarily induce a dynamics on $M_2$. So we impose a tangency
condition along $M_2$, and we obtain a new submanifold
$$
M_3 := \{ p \in M_2 \colon \mbox{there exists } X_p \in \Tan_pM_2 \mbox{ satisfying }\inn(X_p)\Omega_p = \alpha_p \} \, .
$$
A solution $X$ to equation \eqref{pres} does exist in $M_3$ but,
again, such an $X$ is not necessarily tangent to $M_3$, and this
condition must be required. Following this process, we obtain a
sequence of submanifolds
$$
\cdots M_l \hookrightarrow \cdots \hookrightarrow M_2 \hookrightarrow M_1
$$
where the general description of $M_{l+1}$ is
$$
M_{l+1} := \{ p \in M_{l} \colon \mbox{there exists } X_p \in \Tan_pM_l \mbox{ satisfying } \inn(X_p)\Omega_p = \alpha_p \} \, .
$$

If the algorithm terminates at a nonempty set, in the sense that at
some $s\geqslant 1$ we have  $M_{l+1} = M_l$ for every $l \geqslant
s$, then we say that $M_s$ is the final constraint submanifold which
is denoted by $M_f$. It may still happen that $\dim M_f=0$, that is,
$M_f$ is a discrete set of points, and in this case the system does
not admit a proper dynamics. But in the case when $\dim M_f>0$, by
construction, there exists a well-defined solution $X$ of equation
\eqref{pres} along $M_f$.

\subsection{Higher-order tangent bundles}
\label{sec:HOTangentBundles}

In this subsection we recall some basic facts of the higher-order
tangent bundle theory. We particularize our construction to the
case when the configuration space is a Lie group $G$. (See
\cite{proc:Cantrijn_Crampin_Sarlet86,book:DeLeon_Rodrigues85,book:Saunders89} for
details.)

Let $Q$ be a $n$-dimensional smooth manifold, and $k \in \Nat$. The
\textsl{$k$th order tangent bundle} of $Q$, denoted by $\Tan^kQ$, is
the $(k+1)n$-dimensional smooth manifold made of the $k$-jets of
curves $\phi \colon \R \to Q$ with source at $0 \in \R$; that is,
$\Tan^kQ = J_0^k(\R,Q)$. It is a submanifold of $J^k(\R,Q).$ A point
in $\Tan^kQ$ is denoted $j_0^k\phi$, where $\phi$ is a
representative of the equivalence class.

We have the following natural projections: if $r \leqslant k$,
$$
\begin{array}{rcl}
\rho^k_r \colon \Tan^{k}Q & \longrightarrow & \Tan^rQ \\
j^{k}_0\phi & \longmapsto & j^r_0\phi
\end{array} \quad ; \quad
\begin{array}{rcl}
\beta^{k} \colon \Tan^{k}Q & \longrightarrow & Q \\
j^{k}_0\phi & \longmapsto & \phi(0)
\end{array} \, .
$$
Observe that $\rho^k_0 = \beta^k$, where $\Tan^0Q$ is canonically identified with $Q$,
$\rho^s_r \circ \rho^k_s = \rho^k_r$ for every $r \leqslant s \leqslant k$, and
$\rho^k_k = \textnormal{Id}_{\Tan^kQ}$.

From a local chart $(U,\varphi)$ of $Q$, where $\varphi = (\varphi^A)$, $1 \leqslant A \leqslant n$,
the induced local coordinates in $\Tan^kQ$ are constructed as follows:
let $\phi \colon \R \to Q$ be a curve such that $\phi(0) \in U$.
Then, denoting $\phi^A = \varphi^A \circ \phi$,
the point $j^k_0\phi$ is given in $(\beta^{k})^{-1}(U)$ as $(q_0^A,\ldots,q_k^A) \equiv (q_i^A)$,
$0 \leqslant i \leqslant k$, where
$$
q_0^A = \phi^A(0) \quad , \quad
q_i^A = \restric{\frac{d^i\phi^A}{dt^i}}{t=0} \, .
$$
When there is no risk of confusion, we use the standard conventions,
$q_0^A = q^A$, $q_1^A = \dot{q}^A$ and $q_2^A = \ddot{q}^A$.
Using these coordinates, the local expression of the canonical projections are
$$
\rho^k_r(q_0^A,\ldots,q_k^A) = (q_0^A,\ldots,q_r^A) \quad ; \quad
\beta^{k}(q_0^A,\ldots,q_k^A) = (q_0^A) \, .
$$

Now, assume that $Q = G$ is a finite dimensional Lie group, and let us consider
the left-multiplication on itself
$$
\begin{array}{rcl}
G \times G & \longrightarrow & G \\
(g,h) & \longmapsto & gh
\end{array} \, .
$$
If we denote $\pounds_g(h) = gh$ for every $g,h \in G$ the left-translation,
it is obvious that $\pounds_g \colon G \to G$ is a diffeomorphism for every $g \in G$.

\noindent\textbf{Remark:}
The same is valid for the right-translation, but in the sequel we only work with the
left-translation, for the sake of simplicity.

The left-translation enables us to trivialize the tangent and cotangent bundles
of $G$ as follows
$$
\begin{array}{rcl}
\Tan G & \longrightarrow & G \times \g \\
(g,\dot{g}) & \longmapsto & (g,\xi) = (g,g^{-1}\dot{g}) = (g,\Tan_g\pounds_{g^{-1}}\dot{g})
\end{array} \quad ; \quad
\begin{array}{rcl}
\Tan^*G & \longrightarrow & G \times \g^* \\
(g,\alpha_g) & \longmapsto & (g,\alpha) = (g,\Tan_e^*\pounds_g(\alpha_g))
\end{array} \, ,
$$
where $\g = \Tan_eG$ is the Lie algebra of $G$ and $e \in G$ is the neutral element
of the group.

For higher-order tangent bundles, we can also use the left-translation to identify the
$k$th-order tangent bundle of $G$, $\Tan^{k}G$, with $G \times k\g$ as follows: if
$g \colon I \subseteq \R \to G$ is a curve, we define
$$
\begin{array}{rcl}
\Upsilon^{k} \colon \Tan^{k}G & \longrightarrow & G \times k\g \\
j_0^kg & \longmapsto & (g(0), g^{-1}(0)\dot{g}(0), \restric{\frac{d}{dt}}{t=0}(g^{-1}(t)\dot{g}(t)),\ldots,\restric{\frac{d^{k-1}}{dt^{k-1}}}{t=0}(g^{-1}(t)\dot{g}(t)))
\end{array} \, .
$$
It is clear that $\Upsilon^{k}$ is a diffeomorphism.
If we denote by $\xi(t)=g^{-1}(t)\dot{g}(t)$, we can rewrite the above expression as
$$
\Upsilon^{k}(j^k_0g)=(g,\xi^0,\xi^1, \ldots, \xi^{k-1}) \; ,
$$
where
$$
g = g(0) \quad ; \quad \xi^{i} = \restric{\frac{d^i}{dt^i}}{t=0}\xi(t) = \restric{\frac{d^{i}}{dt^{i}}}{t=0}(g^{-1}(t)\dot{g}(t)) \, , \qquad 0 \leqslant i \leqslant k-1 \, .
$$
We will
indistinctly use the notation $\xi^{0}=\xi$,
$\xi^{1}=\dot{\xi}$, where there is no danger of confusion.

In this case, the canonical projections $\rho^{k}_{r}$ and $\beta^{k}$ are denoted by
$$
\begin{array}{rcl}
\tau^k_r \colon \Tan^{k}G & \longrightarrow & \Tan^rG \\
j^{k}_0g & \longmapsto & j^r_0g
\end{array} \quad ; \quad
\begin{array}{rcl}
\tau^{k}_G \colon \Tan^{k}G & \longrightarrow & G \\
j^{k}_0g & \longmapsto & g(0)
\end{array} \, .
$$
Using the previous identifications, we have
$$
\tau^k_r(g,\xi^0,\ldots,\xi^{k-1}) = (g,\xi^0,\ldots,\xi^{r-1}) \quad ; \quad
\tau^{k}_G(g,\xi^0,\ldots,\xi^{k-1}) = g \, .
$$
As before, $\tau^s_r \circ \tau^k_s = \tau^k_r$, $\tau^k_0 = \tau_G^k$,
and $\tau_k^k = \textnormal{Id}_{G}$.

\subsection{Higher-order Hamilton equations in $M \times G$}

Let us consider the manifold $Q = M \times G$,
where $M$ is a $m$-dimensional smooth
manifold and $G$ is a finite dimensional Lie group.
Using the results of Section \ref{sec:HOTangentBundles}, we have
$$
\Tan^*(\Tan^{k-1}Q) = \Tan^*(\Tan^{k-1}(M \times G)) \simeq \Tan^*(\Tan^{k-1}M) \times  G\times (k-1)\g\times k \g^* \, .
$$

In order to geometrically derive Hamilton equations for higher-order variational
problems we need to equip the previous space with a
symplectic structure. Thus, we construct a Liouville $1$-form
$\theta$ and a canonical symplectic 2-form $\omega$ by pull-backing
the canonical Liouville forms in $\Tan^*(\Tan^{k-1}M)$ and $G \times (k-1)\g \times k\g$.
Let $(q_i^A,p_A^i)$ be the natural coordinates in $\Tan^*(\Tan^{k-1}M)$, and
denote by $\bfxi\in (k-1)\g$
and $\bfalpha \in k\g^*$ with components $\bfxi = (\xi^0,\ldots,
\xi^{k-2})$ and $\bfalpha=(\alpha_0,\ldots,\alpha_{k-1})$. Then,
after a straightforward computation. we deduce that
\begin{align*}
&\theta_{(q_{i}^{A},p_{A}^{i},g,\bfxi,  \bfalpha)}(\widetilde{F}_{i}^{A},\widetilde{G}_{A}^{i},\bfxi_1, \bfnu^1) = p_{A}^{i}\widetilde{F}_{i}^{A} + \langle \bfalpha, \bfxi_1\rangle\, ,\\
&\omega_{(q_{i}^{A},p_{A}^{i},g,\bfxi,  \bfalpha)}\left( (\widetilde{F}_{i}^{A},\widetilde{G}_{A}^{i},\bfxi_1, \bfnu^1), (\bar{F}_{i}^{A},\bar{G}_{A}^{i},\bfxi_2, \bfnu^2)\right)
=
\sum_{i=0}^{k-1}\left(\widetilde{F}_{i}^{A}\bar{G}_{A}^{i}-\bar{F}_{i}^{A}\widetilde{G}_{A}^{i}\right)\\[5pt]
&\hspace{250pt} {} + \langle \bfnu^2, \bfxi_1\rangle -\langle \bfnu^1, \bfxi_2\rangle + \langle\alpha_0, [ \xi^{0}_1, \xi^{0}_2]\rangle
\end{align*}
where ${\bfxi}_a\in k\g$ and $\bfnu^a\in k\g^*$,
$a=1, 2$ with components ${\bfxi}_a=(\xi^{i}_a)$ and ${\bfnu}^a=(\nu_{i}^a)$,
$0\leqslant i\leqslant k-1$, where
each component $\xi^{i}_a\in \g$ and $\nu_{i}^a\in \g^*$. Observe that $\alpha_0$ comes from the
identification $\Tan^*G \simeq G\times \g^*$.

Given a Hamiltonian function $H \in \Cinfty(\Tan^*(\Tan^{k-1}(M \times G)))$, we
compute
\begin{align*}
\d H_{(q_{i}^{A},p_{A}^{i},g,\bfxi,\bfalpha)}(\bar{F}_{i}^{A},\bar{G}_{A}^{i},\bfxi_2, \bfnu^2)
&=  \bar{F}_{i}^{A}\derpar{H}{q_{i}^{A}}(q_{i}^{A},p_{A}^{i},g,\bfxi,\bfalpha) + \bar{G}_{A}^{i}\derpar{H}{p_{A}^{i}}(q_{i}^{A},p_{A}^{i},g,\bfxi,\bfalpha) \\
&+ \left\langle \pounds_g^*\derpar{H}{g}(q_{i}^{A},p_{A}^{i},g,\bfxi,\bfalpha) , {\xi}^{0}_2\right\rangle +\sum_{i=0}^{k-2}\left\langle \derpar{H}{\xi^i}(q_{i}^{A},p_{A}^{i},g,\bfxi,\bfalpha),{\xi}^{i+1}_2\right\rangle \\
&+  \left\langle \bfnu^2, \derpar{H}{\bfalpha}(q_{i}^{A},p_{A}^{i},g,\bfxi,\bfalpha)\right\rangle \, .
\end{align*}

Now we can derive Hamilton equations for a higher-order dynamical system
in a trivial principal bundle. First, let us compute the Hamiltonian vector
field $X_H \in \vf(\Tan^*(\Tan^{k-1}M) \times G \times (k-1)\g \times k\g)$
satisfying the geometric equation $\inn(X_H)\omega = \d H$. If $X_H$ is locally
given by $X_H(q_{i}^{A},p_{A}^{i},g,\bfxi,\bfalpha)=(\widetilde{F}_{i}^{A},\widetilde{G}_{A}^{i},{\bfxi}_1,{\bfnu}^1)$,
then the previous equation gives the following system of equations
\begin{align*}
\widetilde{F}_{i}^{A} &= \derpar{H}{p_A^i}(q_{i}^{A},p_{A}^{i},g,\bfxi,  \bfalpha) \, , \\
\widetilde{G}_{A}^{i} &= -\derpar{H}{q_i^A}(q_{i}^{A},p_{A}^{i},g,\bfxi,  \bfalpha) \, , \\
\bfxi_1 &= \derpar{H}{\bfalpha}(q_{i}^{A},p_{A}^{i},g,\bfxi,  \bfalpha) \, ,\\
\nu_{0}^1 &= -\pounds_g^* \derpar{H}{g}(q_{i}^{A},p_{A}^{i},g,\bfxi,  \bfalpha) + ad_{\xi^{0}_1}^*\alpha_{0} \, ,\\
\nu_{i+1}^1 &= - \derpar{H}{\xi^i}(q_{i}^{A},p_{A}^{i},g,\bfxi,  \bfalpha) \, ,\qquad 0\leqslant i\leqslant k-2 \, .
\end{align*}
Finally, if $\gamma(t) = (q_i^A(t),p_A^i(t),g(t),\bfxi(t),\bfalpha(t))$ is an
integral curve of $X_H$, then from the condition $X_H \circ \gamma = \dot{\gamma}$
we obtain the higher-order Hamilton equations
\begin{align*}
\dot{q}_{i}^{A} &= \frac{\partial H}{\partial p_{A}^{i}}(q_{i}^{A},p_{A}^{i},g,\bfxi,  \bfalpha) \, , \qquad 0\leqslant i\leqslant k-1 \, ,\\
\dot{p}_{A}^{i} &= -\frac{\partial H}{\partial q_{i}^{A}}(q_{i}^{A},p_{A}^{i},g,\bfxi,  \bfalpha) \, , \qquad 0\leqslant i\leqslant k-1 \, ,\\
\dot{g} &= g\derpar{H}{\alpha_0}(q_{i}^{A},p_{A}^{i},g,\bfxi,  \bfalpha) \, ,\\
\dot{\xi}^{i} &= \derpar{H}{\alpha_{i}}(q_{i}^{A},p_{A}^{i},g,\bfxi,  \bfalpha) \, , \qquad 1\leqslant i\leqslant k-1\, ,\\
\dot{\alpha}_{0} &= -\pounds_g^*\derpar{H}{g}(q_{i}^{A},p_{A}^{i},g,\bfxi,  \bfalpha) + ad_{{\partial H}/{\partial \alpha_{0}}}^*\alpha_{0}\, ,\\
\dot{\alpha}_{i+1} &= -\derpar{H}{\xi^{i}}(q_{i}^{A},p_{A}^{i},g,\bfxi,  \bfalpha) \, ,\qquad 0\leqslant i\leqslant k-2 \, .
\end{align*}

\section{Geometric formalism for higher-order variational problems}
\label{sec:HOVariationalProblems}

In this section, we  describe the main results of the paper. First,
we intrinsically  derive  the equations of motion for Lagrangian
systems defined on  higher-order trivial principal bundles and,
finally, we extend the results to the case of variationally
constrained problems.

\subsection{Unconstrained problem}
\label{sec:UnconstrainedProblem}

\subsubsection{Geometrical setting}

Let $Q$ be a finite dimensional smooth manifold modeling the
configuration space of a $k$th-order dynamical system, and let
$\Lag \in \Cinfty(\Tan^kQ)$ be a Lagrangian function describing the
dynamics of the system. Consider the Pontryagin bundle $\W =
\Tan^{k}Q \times_{\Tan^{k-1}Q} \Tan^*(\Tan^{k-1}Q)$ in a similar way as in
\cite{art:Prieto_Roman11}. Now, if we take $Q = M \times G$, where
$M$ is a $m$-dimensional smooth manifold and $G$ is a finite
dimensional Lie group, we have
\begin{align*}
\W &= \Tan^{k}(M \times G) \times_{\Tan^{k-1}(M \times G)} \Tan^*(\Tan^{k-1}(M \times G)) \\
& \simeq (\Tan^{k}M \times_{\Tan^{k-1}M} \Tan^*(\Tan^{k-1}M)) \times (\Tan^{k}G \times_{\Tan^{k-1}G} \Tan^*(\Tan^{k-1}G))
= \W_M \times \W_G
\end{align*}
where we denote $\W_M := \Tan^{k}M \times_{\Tan^{k-1}M}
\Tan^*(\Tan^{k-1}M)$ and $\W_G := \Tan^{k}G \times_{\Tan^{k-1}G}
\Tan^*(\Tan^{k-1}G)$. Using left-trivialization and the results in
Section \ref{sec:HOTangentBundles} we have the following identifications
$$
\Tan^{k}G \simeq G \times k\g; \quad  \quad \Tan^*(\Tan^{k-1}G) \simeq G \times (k-1)\g \times k\g^*,
$$
and therefore the manifold $\W_G$ admits the identification
$$
\W_G \simeq G \times k\g \times k\g^*.
$$
Taking into account all the previous comments, we can consider the
diagram illustrating the situation {\small
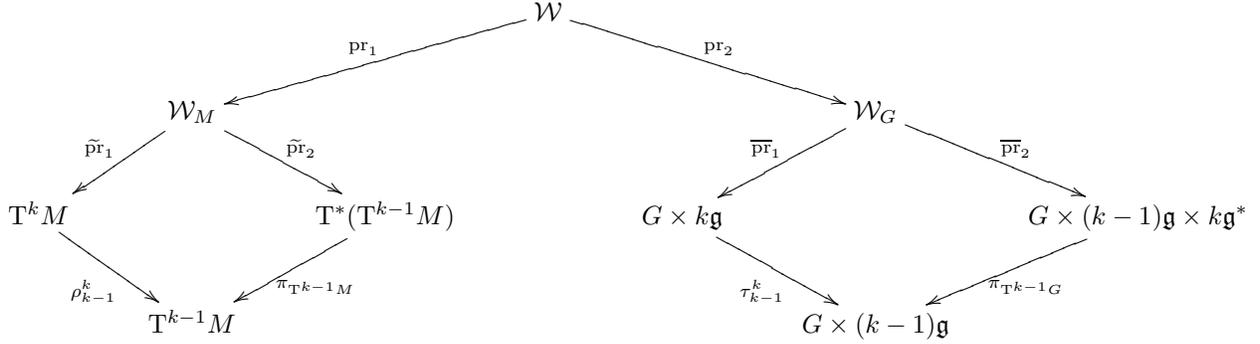
$$
\xymatrix{
\ & \ & \ & \W \ar[dll]_{\pr_1} \ar[drr]^{\pr_2} & \ & \ & \ \\
\ & \W_M \ar[dl]_{\widetilde{\pr}_1} \ar[dr]^{\widetilde{\pr}_2} & \ &  & \ & \W_G \ar[dl]_{\overline{\pr}_1} \ar[dr]^{\overline{\pr}_2} & \ \\
\Tan^kM \ar[dr]_{\rho^{k}_{k-1}} &  & \Tan^*(\Tan^{k-1}M) \ar[dl]^{\pi_{\Tan^{k-1}M}} & \ & G \times k\g \ar[dr]_{\tau^{k}_{k-1}} & \ & G \times (k-1)\g \times k\g^* \ar[dl]^{\pi_{\Tan^{k-1}G}} \\
\ & \Tan^{k-1}M & \ &  & \ & G \times (k-1)\g & \
}
$$}
where all the maps are the canonical projections.

Let $(q_i^A,q_k^A,p_A^i)$, where $0 \leqslant i \leqslant k-1$ and
$1 \leqslant A \leqslant m$, be a set of local coordinates in $\W_M$
(see \cite{art:Prieto_Roman11} for details), and
$(g,\bfxi,\xi^{k-1},\bfalpha)$, where $\bfxi =
(\xi^0,\ldots,\xi^{k-2}) \in (k-1)\g$ and $\bfalpha =
(\alpha_0,\ldots,\alpha_{k-1}) \in k\g^*$, a set of local
coordinates in $\W_G$ (see \cite{art:Colombo_Martin11} for details).
Then, the induced natural coordinates in $\W = \W_M \times \W_G$ are
$(q_i^A,q_k^A,p_A^i,g,\bfxi,\xi^{k-1},\bfalpha)$. Using these
coordinates, the above projections have the following local
expressions
\begin{align*}
\pr_1(q_i^A,q_k^A,p_A^i,g,\bfxi,\xi^{k-1},\bfalpha) = (q_i^A,q_k^A,p_A^i) \quad &; \quad \pr_2(q_i^A,q_k^A,p_A^i,g,\bfxi,\xi^{k-1},\bfalpha) = (g,\bfxi,\xi^{k-1},\bfalpha) \\
\widetilde{\pr}_1(q_i^A,q_k^A,p_A^i) = (q_i^A,q_k^A) \quad &; \quad \widetilde{\pr}_2(q_i^A,q_k^A,p_A^i) = (q_i^A,p_A^i) \\
\overline{\pr}_1(g,\bfxi,\xi^{k-1},\bfalpha) = (g,\bfxi,\xi^{k-1}) \quad &; \quad \overline{\pr}_2(g,\bfxi,\xi^{k-1},\bfalpha) = (g,\bfxi,\bfalpha)
\end{align*}

The bundle $\W$ is endowed with some canonical geometric structures. First, let $\omega_{k-1} \in \df^{2}(\Tan^*(\Tan^{k-1}M))$
and $\omega_{G\times(k-1)\g} \in \df^{2}(G \times (k-1)\g \times k\g^*)$ be the canonical symplectic forms
in $\Tan^*(\Tan^{k-1}M)$ and $\Tan^*(\Tan^{k-1}G) \simeq G \times (k-1)\g \times k\g^*$, respectively.
Then, we can consider the presymplectic forms $\Omega_M = \widetilde{\pr}_2^*\,\omega_{k-1} \in \df^{2}(\W_M)$
and $\Omega_G = \overline{\pr}_2^*\,\omega_{G \times (k-1)\g} \in \df^{2}(\W_G)$. Then we define
the following presymplectic form in $\W$
\begin{equation}\label{eqn:PresymplecticFormDef}
\Omega = \pr_1^*\Omega_M + \pr_2^*\Omega_G \in \df^{2}(\W) \, .
\end{equation}
Observe that since $\ker\Omega_M = \vf^{V(\widetilde{\pr}_2)}(\W_M)$ and $\ker\Omega_G = \vf^{V(\overline{\pr}_2)}(\W_G)$,
we have
$$
\ker\Omega = \vf^{V(\widetilde{\pr}_2 \circ \pr_1)}(\W) \cap \vf^{V(\pr_2)}(\W) + \vf^{V(\overline{\pr}_2 \circ \pr_2)}(\W)  \cap \vf^{V(\pr_1)}(\W) \, .
$$
In natural coordinates, recall that the forms $\omega_{k-1}$ and $\omega_{G\times(k-1)\g}$ are locally given by
\begin{align*}
(\omega_{k-1})_{(q_i^A,p_A^i)} \left( (\tilde{F}_i^A,\tilde{G}_A^i),(\bar{F}_i^A,\bar{G}_A^i) \right) &= \sum_{i=0}^{k-1} \left( \tilde{F}_i^A\bar{G}_A^i - \bar{F}_i^A\tilde{G}_A^i \right) \\
(\omega_{G\times(k-1)\g})_{(g,\bfxi,\bfalpha)}\left( (\bfxi_1,\bfnu^1),(\bfxi_2,\bfnu^2) \right)
&= \langle \bfnu^2,\bfxi_1 \rangle - \langle \bfnu^1,\bfxi_2 \rangle + \langle \alpha_0, [\xi_1^0,\xi_2^0] \rangle \\
&= \sum_{i=0}^{k-1} \left( \langle \nu_i^2,\xi_1^i \rangle -\langle \nu_i^1,\xi_2^i \rangle \right) + \langle \alpha_0, [\xi_1^0,\xi_2^0] \rangle
\end{align*}
where $\bfxi_a = (\xi_a^0,\ldots,\xi_a^{k-1}) \in k\g$, $\bfnu^a = (\nu_0^a,\ldots,\nu_{k-1}^a) \in k\g^*$ ($a = 1,2$),
$(q_i^A,p_A^i) \in \Tan^*(\Tan^{k-1}M)$ and $(\tilde{F}_i^A,\tilde{G}_A^i),(\bar{F}_i^A,\tilde{G}_A^i) \in \Tan_{(q_i^A,p_A^i)}\Tan^*(\Tan^{k-1}M)$.
Therefore, the presymplectic form $\Omega \in \df^{2}(\W)$ is locally given by
\begin{equation}\label{eqn:OmegaLocal}
\Omega_{(q_i^A,q_k^A,p_A^i,g,\bfxi,\xi^{k-1},\bfalpha)}\left( X_1, X_2 \right)
= \sum_{i=0}^{k-1} \left( \tilde{F}_i^A\bar{G}_A^i - \bar{F}_i^A\tilde{G}_A^i + \langle \nu_i^2,\xi_1^i \rangle - \langle \nu_i^1,\xi_2^i \rangle \right) + \langle \alpha_0, [\xi_1^0,\xi_2^0] \rangle \, ,
\end{equation}
where $X_1 = (\tilde{F}_i^A,\tilde{F}_k^A,\tilde{G}_A^i,\bfxi_1,\xi_1^{k},\bfnu^1),
X_2 = (\bar{F}_i^A,\bar{F}_k^A,\bar{G}_A^i,\bfxi_2,\xi_2^{k},\bfnu^2) \in \vf(\W)$.
Moreover, a local basis for $\ker\Omega$ is
\begin{equation}\label{eqn:KernelOmegaLocal}
\ker\Omega = \left\langle \derpar{}{q_k^A} \, , \, \derpar{}{\xi^{k-1}} \right\rangle \, .
\end{equation}

The second relevant canonical structure in $\W$ is the
\textsl{coupling function} $\C \in \Cinfty(\W)$.
First, since $\Tan^kM$
is canonically embedded into $\Tan(\Tan^{k-1}M)$, we can define a canonical pairing
between the elements of $\Tan^*(\Tan^{k-1}M)$ and the elements of $\Tan^{k}M$
as a function in $\Cinfty(\W_M)$. Indeed, let $p \in \Tan^{k}M$ be a point,
$q = \rho_{k-1}^{k}(p)$ its projection to $\Tan^{k-1}M$ and $\alpha_q \in \Tan^*_q(\Tan^{k-1}M)$
a covector. Then, the function $\C_M \in \Cinfty(\W_M)$ is defined as
$$
\begin{array}{rcl}
\C_M \colon \Tan^{k}M \times_{\Tan^{k-1}M} \Tan^*(\Tan^{k-1}M) & \longrightarrow & \R \\
(p,\alpha_q) & \longmapsto & \langle \alpha_q , j_k(p)_q \rangle_{k-1}
\end{array} \ ,
$$
where $j_k \colon \Tan^kM \hookrightarrow \Tan(\Tan^{k-1}M)$ is the canonical embedding,
$j_k(p)_q \in \Tan_q(\Tan^{k-1}M)$ the corresponding tangent vector to $\Tan^{k-1}M$ in $q$, and
$\langle \cdot\,,\,\cdot \rangle_{k-1} \colon \Tan(\Tan^{k-1}M) \times \Tan^*(\Tan^{k-1}M) \to \R$
the canonical pairing. In natural coordinates, if $p = (q_0^A,\ldots,q_k^A)$, then
$q = \rho^k_{k-1}(p) = (q_0^A,\ldots,q_{k-1}^A)$, and the canonical embedding
is locally given by $j_k(p) = (q_0^A,\ldots,q_{k-1}^A,q_1^A,\ldots,q_k^A)$.
Hence, if $j_k(p)_q = q_{i+1}^A\restric{\derpar{}{q_i^A}}{q} \in \Tan_q(\Tan^{k-1}M)$
and $\alpha_q = p_A^i\restric{\d q_i^A}{q} \in \Tan^*_q(\Tan^{k-1}M)$, then the local expression
of the function $\C_M$ is
$$
\C_M(q_i^A,q_k^A,p_A^i) = \restric{p_A^iq_{i+1}^A}{q} \, .
$$

On the other hand, we can define a canonical pairing in $\W_G \simeq G \times k\g \times k\g^*$
as a function $\C_G \in \Cinfty(\W_G)$ as follows
$$
\begin{array}{rcl}
\C_G \colon G \times k\g \times k\g^* & \longrightarrow & \R \\
(g,\bfxi,\xi^{k-1},\bfalpha) & \longmapsto & \langle \alpha_i,\xi^i \rangle
\end{array} \ ,
$$

Bearing in mind the above constructions, we can give the following definition.

\begin{definition}
The \textnormal{coupling function} $\C \in \Cinfty(\W)$ is defined as
\begin{equation}\label{eqn:CouplingFunctionDef}
\C = \pr_1^*\C_M + \pr_2^*\C_G \, .
\end{equation}
\end{definition}

In the induced natural coordinates of $\W$, bearing in mind the
local expressions of both $\C_M$ and $\C_G$, and the coordinate
expressions of the projections $\pr_1$ and $\pr_2$, we have that the
coupling function $\C \in \Cinfty(\W)$ is locally given by
$$
\C(q_i^A,q_k^A,p_A^i,g,\bfxi,\xi^{k-1},\bfalpha) = \sum_{i=0}^{k-1} \left( p_A^iq_{i+1}^A + \langle \alpha_i, \xi^i \rangle \right)  \, .
$$

Now, given a Lagrangian function $\Lag \in \Cinfty(\Tan^k(M \times
G))$, we can define the \textsl{Hamiltonian function} $H \in
\Cinfty(\W)$ as
\begin{equation}\label{eqn:HamiltonianFunctionDef}
H = \C - \pi^*\Lag \, ,
\end{equation}
where $\pi \colon \W \to \Tan^k(M \times G)$ is the natural projection, and whose local expression is
$$
H(q_i^A,q_k^A,p_A^i,g,\bfxi,\xi^{k-1},\bfalpha) = \sum_{i=0}^{k-1} \left( p_A^iq_{i+1}^A + \langle \alpha_i, \xi^i \rangle \right) - \Lag(q_i^A,q_k^A,g,\bfxi,\xi^{k-1}) \, .
$$

\subsubsection{Dynamical equation}
\label{sec:DynEqUnconstrained}

The dynamical equation for a presymplectic Hamiltonian system $(\W,\Omega,H)$ is geometrically written as
\begin{equation}\label{eqn:DynEq}
\inn(X) \Omega = \d H \ , \quad \mbox{for } X \in \vf(\W) \, .
\end{equation}
Then, following \cite{art:Gotay_Nester_Hinds78}, we have
\begin{proposition}\label{prop:CompatibilitySubmanifold}
A solution $X \in \vf(\W)$ to equation \eqref{eqn:DynEq} exists only on the points of the submanifold
$\W_c \stackrel{j_c}{\hookrightarrow} \W$ defined by
$$
\W_c = \left\{ p \in \W \colon (\inn(Y)\d H)(p) = 0 \, , \ \forall \, Y \in \ker\Omega \right\} \ .
$$
\end{proposition}

In natural coordinates, since $\d H \in \df^{1}(\W)$ is locally given by
\begin{equation}\label{eqn:DiffHamiltonianLocal}
\begin{array}{rl}
\d H_{(q_i^A,q_k^A,p_A^i,g,\bfxi,\xi^{k-1},\bfalpha)}(Y) =
&  -\tilde{F}_0^A\derpar{\Lag}{q_0^A} + \tilde{F}_{i+1}^A\left( p_A^i - \derpar{\Lag}{q_{i+1}^A} \right) + \tilde{G}_A^iq_{i+1}^A \\
\ & + \left\langle - \pounds_g^* \derpar{\Lag}{g} , \xi_2^0 \right\rangle + \left\langle \alpha_i - \derpar{\Lag}{\xi^i} , \xi^{i+1}_2 \right\rangle
+ \langle \bfnu^2,\bfxi \rangle \, ,
\end{array}
\end{equation}
where $Y = (\tilde{F}_i^A,\tilde{F}_k^A,\tilde{G}_A^i,\bfxi_2,\xi_2^k,\bfnu^2) \in \vf(\W)$,
and $\ker\Omega$ has local basis \eqref{eqn:KernelOmegaLocal}, we have
$$
\inn(Y)\d H =
\begin{cases}
p_A^{k-1} - \derpar{\Lag}{q_k^A} \, , & \mbox{if } Y = \derpar{}{q_k^A} \, , \\[10pt]
\alpha_{k-1} - \derpar{\Lag}{\xi^{k-1}} \, , & \mbox{if } Y = \derpar{}{\xi^{k}} \, .
\end{cases}
$$
Therefore, $\W_c \hookrightarrow \W$ is locally defined by the constraints
$$
p_A^{k-1} - \derpar{\Lag}{q_k^A} = 0 \quad ; \quad
\alpha_{k-1} - \derpar{\Lag}{\xi^{k-1}} = 0 \, .
$$

Now, let us compute the local expression of equation \eqref{eqn:DynEq}.
Let $X \in \vf(\W)$ be a generic vector field locally given by
\begin{equation}\label{eqn:GenericVectorField}
X =
F_i^A\derpar{}{q_i^A} + F_k^A\derpar{}{q_k^A} + G_A^i\derpar{}{p_A^i}
+ \xi_1^0\derpar{}{g} + \xi_1^{i+1}\derpar{}{\xi^{i}} + \nu_i^1\derpar{}{\alpha_i}
= (F_i^A,F_k^A,G_A^i,\bfxi_1,\xi_1^{k},\bfnu^1) \, .
\end{equation}
Then, using \eqref{eqn:OmegaLocal} and
\eqref{eqn:DiffHamiltonianLocal}, we have the following system of
equations
\begin{align}
F_i^A = q_{i+1}^A \, , \label{eqn:DynEqLocalHolonomyManifold}\\
G_A^{0} = \derpar{\Lag}{q_0^A} \quad , \quad G_A^i = \derpar{\Lag}{q_i^A} - p_A^{i-1} \, , \label{eqn:DynEqLocalManifold} \\
p_A^{k-1} - \derpar{\Lag}{q_k^A} = 0 \, , \label{eqn:DynEqLocalHOMomentaManifold} \\
\xi_1^{i} = \xi^{i} \, , \label{eqn:DynEqLocalHolonomyGroup} \\
\nu_0^1 = \pounds_g^*\derpar{\Lag}{g} + ad_{\xi_1^0}^*\alpha_0 \quad , \quad \nu_{i+1}^1 = \derpar{\Lag}{\xi^i} - \alpha_i \, , \label{eqn:DynEqLocalGroup} \\
\alpha_{k-1} - \derpar{\Lag}{\xi^{k-1}} = 0 \, . \label{eqn:DynEqLocalHOMomentaGroup}
\end{align}

Therefore, the vector field $X$ solution to equation
\eqref{eqn:DynEq} is locally given by
\begin{align*}
X &= q_{i+1}^A\derpar{}{q_i^A} + F_k^A\derpar{}{q_k^A} + \derpar{\Lag}{q_0^A}\derpar{}{p_A^0} + \left( \derpar{\Lag}{q_i^A} - p_A^{i-1} \right) \derpar{}{p_A^i} \\
&\quad {}+ \xi^0\derpar{}{g} + \xi^{i+1}\derpar{}{\xi^{i}} + \xi_1^{k} \derpar{}{\xi^{k-1}} + \left( \pounds^*_g\derpar{\Lag}{g} + ad^*_{\xi^0}\alpha_0 \right) \derpar{}{\alpha_0}
+ \left( \derpar{\Lag}{\xi^i} - \alpha_i \right) \derpar{}{\alpha_{i+1}} \, .
\end{align*}
Observe that equations \eqref{eqn:DynEqLocalHOMomentaManifold} and
\eqref{eqn:DynEqLocalHOMomentaGroup} are compatibility conditions
that say that the vector field $X$ exists with support on a
submanifold defined locally by these equations. Hence, we recover in
coordinates the result stated in Proposition
\ref{prop:CompatibilitySubmanifold}.

The coefficients $F_k^A$ and $\xi^{k}_1$ are yet to be determined.
Nevertheless, recall that $X$ is a vector field in $\W$ that exists
at support on $\W_c$. Hence, we must study the tangency of $X$ along
the submanifold $\W_c$; that is, we must require
$\restric{\Lie(X)\zeta}{\W_c} \equiv \restric{X(\zeta)}{\W_c} = 0$ for every constraint function
$\zeta$ defining $\W_c$. Thus,
taking into account that $\W_c$ is locally defined by equations
\eqref{eqn:DynEqLocalHOMomentaManifold} and
\eqref{eqn:DynEqLocalHOMomentaGroup}, the tangency condition for $X$
along $\W_c$ gives the following equations
\begin{equation}\label{eqn:TangencyCondition}
\begin{array}{l}
\derpar{\Lag}{q_{k-1}^A} - p_A^{k-2} = q_{i+1}^B \derpars{\Lag}{q_i^B}{q_k^A} + F_k^B\derpars{\Lag}{q_k^B}{q_k^A} \\[12pt]
\hspace{85pt} + \, \xi^0\pounds_g^*\derpars{\Lag}{g}{q_k^A} + \xi^{i+1}\derpars{\Lag}{\xi^i}{q_k^A} + \xi_1^k\derpars{\Lag}{\xi^{k-1}}{q_k^A} \, , \\[15pt]
\derpar{\Lag}{\xi^{k-2}} - \alpha_{k-2} = q_{i+1}^B\derpars{\Lag}{q_i^B}{\xi^{k-1}} + F_k^B\derpars{\Lag}{q_k^B}{\xi^{k-1}} \\[12pt]
\hspace{85pt} \, + \xi^0\pounds_g^*\derpars{\Lag}{g}{\xi^{k-1}} + \xi^{i+1}\derpars{\Lag}{\xi^i}{\xi^{k-1}} + \xi^k_1 \derpars{\Lag}{\xi^{k-1}}{\xi^{k-1}} \, .
\end{array}
\end{equation}
These equations enable us to determinate the remaining coefficients
$F_k^A$ and $\xi_1^k$ of the vector field $X$. Observe that if the
Hessian matrix of $\Lag$ with respect to the highest-order
``velocities'', $q_k^A$ and $\xi^{k-1}$, is invertible, that is,
$$
\det
\begin{pmatrix}
\derpars{\Lag}{q_k^B}{q_k^A} & \derpars{\Lag}{q_k^A}{\xi^{k-1}} \\[12pt]
\derpars{\Lag}{\xi^{k-1}}{q_k^A} & \derpars{\Lag}{\xi^{k-1}}{\xi^{k-1}}
\end{pmatrix}
(p)\neq 0 \ , \quad \mbox{for every } p \in \Tan^{k}M \times G \times k\g \, ,
$$
then the previous system of equations has a unique solution for
$F_k^A$ and $\xi_1^k$, thus obtaining a unique vector field $X \in
\vf(\W)$ solution to the equation \eqref{eqn:DynEq}. In particular,
the constraint algorithm finishes at the first step. Otherwise, new
constraints may arise from equations \eqref{eqn:TangencyCondition},
and the algorithm continues if necessary.

\noindent\textbf{Remark:} In the particular case when the Hessian
matrix of the Lagrangian function is a block diagonal matrix, that
is,
$$
\derpars{\Lag}{q_i^A}{g} = 0 \quad \mbox{and} \quad \derpars{\Lag}{q_i^A}{\xi^j} = 0 \ , \
\mbox{for every } 1 \leqslant A \leqslant m \, , \, 0 \leqslant i \leqslant k \, , \, 0 \leqslant j \leqslant k-1 \, ,
$$
then equations \eqref{eqn:TangencyCondition} become
$$
\begin{array}{l}
\derpar{\Lag}{q_{k-1}^A} - p_A^{k-2} = q_{i+1}^B \derpars{\Lag}{q_i^B}{q_k^A} + F_k^B\derpars{\Lag}{q_k^B}{q_k^A} \, , \\[10pt]
\derpar{\Lag}{\xi^{k-2}} - \alpha_{k-2} = \xi^0\pounds_g^*\derpars{\Lag}{g}{\xi^{k-1}} + \xi^{i+1}\derpars{\Lag}{\xi^i}{\xi^{k-1}} + \xi^k_1 \derpars{\Lag}{\xi^{k-1}}{\xi^{k-1}}
\end{array}
$$
In this case, to solve the equation \eqref{eqn:DynEq} in $\W$ is
equivalent to solve separately the corresponding equations in $\W_M$
and $\W_G$ following the patterns in \cite{art:Prieto_Roman11} and
\cite{art:Colombo_Martin11}, respectively, and then take $X = X_M +
X_G$ as a solution of equation \eqref{eqn:DynEq}, where $X_M \in
\ker\pr_2$ is a vector field $\pr_1$-related with the vector field
solution to the equation in $\W_M$ and $X_G \in \ker\pr_1$ is a
vector field $\pr_2$-related with the solution of the equation in
$\W_G$.

Now, let $\gamma \colon \R \to \W$ be an integral curve of $X$
locally given by
\begin{equation}\label{eqn:GenericCurve}
\gamma(t) = (q_i^A(t),q_k^A(t),p_A^i(t),g(t),\xi^i(t),\alpha_{i}(t)) \, .
\end{equation}
From the condition $X \circ \gamma = \dot{\gamma}$ we obtain the
following system of differential equations for the component functions of
$\gamma$
\begin{align}
\dot{q}_i^A = q_{i+1}^A \, , \\
\dot{p}_A^0 = \derpar{\Lag}{q_0^A} \quad , \quad \dot{p}_A^i = \derpar{\Lag}{q_i^A} - p_A^{i-1} \, ,\label{eqn:DynEqCurveLocalManifold}\\
\dot{g} = g\xi^0 \, \quad , \quad \dot{\xi}^{i-1} = \xi^{i} \, , \\
\dot{\alpha}_0 = \pounds_g^*\derpar{\Lag}{g} + ad_{\xi_1^0}^*\alpha_0 \quad , \quad \dot{\alpha}_{i+1} = \derpar{\Lag}{\xi^i} - \alpha_i \label{eqn:DynEqCurveLocalGroup} \, ,
\end{align}
in addition to equations \eqref{eqn:DynEqLocalHOMomentaManifold} and
\eqref{eqn:DynEqLocalHOMomentaGroup}. Now, using equations
\eqref{eqn:DynEqLocalHOMomentaManifold} in combination with
equations \eqref{eqn:DynEqCurveLocalManifold} we obtain the
$k$th-order Euler-Lagrange equations
\begin{equation}\label{eqn:EulerLagrange}
\sum_{i=0}^{k}(-1)^{i}\restric{\frac{d^i}{dt^i}\derpar{\Lag}{q_i^A}}{\gamma} = 0 \, .
\end{equation}
On the other hand, using equations
\eqref{eqn:DynEqLocalHOMomentaGroup} in combination with equations
\eqref{eqn:DynEqCurveLocalGroup} we obtain the $k$th-order
trivialized Euler-Lagrange equations
\begin{equation}\label{eqn:TrivializedEulerLagrange}
\left( \frac{d}{dt} - ad_{\xi^0}^* \right) \sum_{i=0}^{k-1}(-1)^{i} \restric{\frac{d}{dt^i}\derpar{\Lag}{\xi^{i}}}{\gamma} = \pounds_g^*\restric{\left( \derpar{\Lag}{g} \right)}{\gamma} \, .
\end{equation}
Therefore, a dynamical trajectory $\gamma \colon \R \to \W$ of the
system must satisfy the following local equations
$$
\sum_{i=0}^{k}(-1)^{i}\restric{\frac{d^i}{dt^i}\derpar{\Lag}{q_i^A}}{\gamma} = 0 \quad , \quad
\left( \frac{d}{dt} - ad_{\xi^0}^* \right) \sum_{i=0}^{k-1}(-1)^{i} \restric{\frac{d}{dt^i}\derpar{\Lag}{\xi^{i}}}{\gamma} = \pounds_g^*\restric{\left( \derpar{\Lag}{g} \right)}{\gamma} \, .
$$

\noindent\textbf{Remark:} The above equations may be compatible or
not. If not, a constraint algorithm must be used in order to find a
final submanifold where the above equations have a solution
(if such a submanifold exists).

Finally, if the Lagrangian function $\Lag \in \Cinfty(\Tan^{k}M
\times G \times k\g)$ is left-invariant, that is,
$$
\Lag(q_i^A,q_k^A,g,\xi^{i}) = \Lag(q_i^A,q_k^A,h,\xi^{i}) \, ,
$$
for all $g,h \in G$, then we can define the reduced Lagrangian $\ell
\in \Cinfty(\Tan^kM \times k\g)$ by
$$
\ell(q_i^A,q_k^A,\xi^{i}) = \Lag(q_i^A,q_k^A,e,\xi^{i}) \, ,
$$
and therefore equations \eqref{eqn:TrivializedEulerLagrange} become
the $k$th order Euler-Poincar\'{e} equations
\begin{equation}\label{eqn:EulerPoincare}
\left( \frac{d}{dt} - ad_{\xi^0}^* \right) \sum_{i=0}^{k-1}(-1)^{i} \restric{\frac{d}{dt^i}\derpar{\ell}{\xi^{i}}}{\gamma} = 0 \, .
\end{equation}
Observe that equations \eqref{eqn:EulerLagrange} remain the same
with the reduced Lagrangian function, just replacing $\Lag$ by
$\ell$.

Equations \eqref{eqn:EulerLagrange} and
\eqref{eqn:TrivializedEulerLagrange} are exactly the same
that one of the authors derived in
\cite{art:Colombo_Jimenez_Martin12} using variational techniques.
Our derivation allows us to identify staightforwardly the geometric
preservation of the system, for instance, preservation of the
Hamiltonian or (pre)symplecticity of the flow.

\subsubsection{A theoretical example}

Now, we give a theoretical example inspired by the applications
in Clebsch variational principle and continuum mechanics studied in
\cite{art:GayBalmaz_Holm_Ratiu13,book:Holm08}.

Let us consider the particular case when the manifold $M$ is the
dual of a real vector space, that is, $M = V^*$, where $V$ is a
finite dimensional real vector space. In this case we have the
following identifications
$$
\begin{array}{c}
\Tan^kV^* \simeq (k+1)V^* \quad , \quad
\Tan^*V^* \simeq V^*\times V \\[5pt]
\Tan^*(\Tan^{k-1}V^*) \simeq k(V^*\times V)
\end{array}
$$
Using these identifications, we have
$$
\W_{V^*} = \Tan^kV^* \times_{\Tan^{k-1}V^*} \Tan^*(\Tan^{k-1}V^*)
\simeq (k+1)V^* \times kV
$$
Since $V$ and $V^*$ have global charts of coordinates defined by any basis, we will denote an element
of $(k+1)V^*$ by $(\bfmu,\mu_k) \equiv (\mu_i,\mu_k) \equiv (\mu_0,\ldots,\mu_k)$, where $\mu_j \in V^*$
for every $0 \leqslant j \leqslant k$, and an element of $kV$ will be denoted by
$(\bfv) \equiv (v^0,\ldots,v^{k-1})$, where $v^j \in V$ for every $0 \leqslant j \leqslant k-1$.
Then, an element of $\W_{V^*}$ will be denoted $(\bfmu,\mu_k,\bfv) \equiv (\mu_0,\ldots,\mu_{k-1},\mu_k,v^0,\ldots,v^{k-1})$.

The canonical projections $\widetilde{\pr}_1 \colon (k+1)V^*\times
kV \to (k+1)V^*$ and $\widetilde{\pr}_2 \colon (k+1)V^* \times kV
\to kV^* \times kV$ are given by
$$
\widetilde{\pr}_1(\bfmu,\mu_k,\bfv) = (\bfmu,\mu_k) \quad ; \quad \widetilde{\pr}_2(\bfmu,\mu_k,\bfv) = (\bfmu,\bfv) \, .
$$

Let $\omega_{kV^*} \in \df^{2}(k(V^*\times V))$ be the canonical
symplectic form in $k(V^*\times V)$, which is given by
$$
(\omega_{kV^*})_{(\bfmu,\bfv)}((\beta_i^1,u_1^i),(\beta_i^2,u_2^i)) =
\sum_{i=0}^{k-1} ( \langle u_2^i,\beta^1_i \rangle_{V^*} - \langle u_1^i,\beta^2_i \rangle_{V^*})
= \sum_{i=0}^{k-1} (\beta_i^1(u^i_2) - \beta_i^2(u_1^i)) \, .
$$
where $\langle \cdot , \cdot \rangle_{V^*}$ is the canonical pairing
between the elements of $V^*$ and its dual $V^{**} \simeq V$. We
define the presymplectic form in $\W_{V^*}$ as $\Omega_{V^*} =
\widetilde{\pr}_2^*\omega_{kV^*} \in \df^{2}(\W_{V^*})$. This
$2$-form is given locally by
$$
(\Omega_{V^*})_{(\bfmu,\mu_k,\bfv)}((\beta_i^1,\beta_k^1,u_1^i),(\beta_i^2,\beta_k^2,u_2^i)) =
\sum_{i=0}^{k-1} (\beta_i^1(u^i_2) - \beta_i^2(u_1^i)) \, .
$$

Now, we define the canonical pairing in $\W_{V^* } \simeq (k+1)V^*
\times kV$ as a function $\C_{V^*} \in \Cinfty(\W_{V^*})$ as
follows:
$$
\begin{array}{rcl}
\C_{V^*} \colon (k+1)V^* \times kV & \longrightarrow & \R \\
(\bfmu,\mu_k,\bfv) & \longmapsto & \langle v^i, \mu_{i+1} \rangle_{V^*} = \mu_{i+1}(v^i)
\end{array}
\, .
$$

With these elements, we can follow the patterns in Section
\ref{sec:UnconstrainedProblem}. Let us consider the manifold $Q =
V^*\times G.$ Then we consider the Pontryagin bundle
$$
\W \simeq (k+1)V^* \times kV \times G \times k\g \times k\g^* \, .
$$
Then, the diagram in Section \ref{sec:UnconstrainedProblem} becomes
{\footnotesize
$$
\xymatrix{
\ & \ & \ & \W \ar[dll]_{\pr_1} \ar[drr]^{\pr_2} & \ & \ & \ \\
\ & (k+1)V^* \times kV \ar[dl]_{\widetilde{\pr}_1} \ar[dr]^{\widetilde{\pr}_2} & \ & \ & \ & G \times k\g \times k\g^* \ar[dl]_{\overline{\pr}_1} \ar[dr]^{\overline{\pr}_2} & \ \\
(k+1)V^* \ar[dr]_{\rho^{k}_{k-1}} &  & kV^* \times kV \ar[dl]^{\pi_{\Tan^{k-1}V^*}} & \ & G \times k\g \ar[dr]_{\tau^{k}_{k-1}} & \ & G \times (k-1)\g \times k\g^* \ar[dl]^{\pi_{\Tan^{k-1}G}} \\
\ & kV^* & \ & \ & \  & G \times (k-1)\g & \
}
$$
}

Let $(\bfmu,\mu_k,\bfv)$ be local coordinates in $\W_{V^*}$ and
$(g,\bfxi,\xi^k,\bfalpha)$ are local coordinates in $\W_G$. Then,
the induced local coordinates in $\W$ are
$(\bfmu,\mu_k,\bfv,g,\bfxi,\xi^k,\bfalpha)$. The canonical
projections $\pr_1 \colon (k+1)V^* \times
kV \times G \times k\g \times k\g^* \to (k+1)V^* \times kV$ and
$\pr_2 \colon (k+1)V^* \times kV \times G \times k\g \times k\g^*
\to G \times k\g \times k\g^*$ are given in these coordinates by
$$
\pr_1(\bfmu,\mu_k,\bfv,g,\bfxi,\xi^k,\bfalpha) = (\bfmu,\mu_k,\bfv) \quad ; \quad
\pr_2(\bfmu,\mu_k,\bfv,g,\bfxi,\xi^k,\bfalpha) = (g,\bfxi,\xi^k,\bfalpha).
$$

The presymplectic form $\Omega \in \df^{2}(\W)$ defined in \eqref{eqn:PresymplecticFormDef}
is now given by
\begin{equation*}
\Omega_{(\bfmu,\mu_k,\bfv,g,\bfxi,\xi^{k-1},\bfalpha)}\left( X_1, X_2 \right)
= \sum_{i=0}^{k-1} \left(\beta_i^1(u^i_2) - \beta_i^2(u_1^i) + \langle \nu_i^2,\xi_1^i \rangle - \langle \nu_i^1,\xi_2^i \rangle \right) + \langle \alpha_0, [\xi_1^0,\xi_2^0] \rangle \, ,
\end{equation*}
where $X_1 = (\beta_i^1,\beta_k^1,u_1^i,\bfxi_1,\xi_1^{k},\bfnu^1),
X_2 = (\beta_i^2,\beta_k^2,u_2^i,\bfxi_2,\xi_2^{k},\bfnu^2) \in
\vf(\W)$. The coupling function $\C \in \Cinfty(\W)$ is
$$
\C(\bfmu,\mu_k,\bfv,g,\bfxi,\xi^{k-1},\bfalpha) = \sum_{i=0}^{k-1} \left( \mu_{i+1}(v^i) + \langle \alpha_i, \xi^i \rangle \right)  \, .
$$

Let $\Lag \in \Cinfty((k+1)V^* \times G \times k\g)$ be a
$k$th-order Lagrangian function. We define the Hamiltonian function
$H \in \Cinfty(\W)$ by
$$
H(\bfmu,\mu_k,\bfv,g,\bfxi,\xi^{k-1},\bfalpha) = \sum_{i=0}^{k-1} \left( \mu_{i+1}(v^i) + \langle \alpha_i, \xi^i \rangle \right) - \Lag(\bfmu,\mu_k,g,\bfxi,\xi^{k-1}) \, .
$$

Now, if $X \in \vf(\W)$ is a vector field, locally given by
$$
X = \beta_i\derpar{}{\mu_i} + \beta_k\derpar{}{\mu_k} + u^i\derpar{}{v^i}
+ \xi_1^0\derpar{}{g} + \xi_1^{i+1}\derpar{}{\xi^{i}} + \nu_i^1\derpar{}{\alpha_i}
= (\beta_i,\beta_k,u^i,\xi_1^i,\xi_1^{k},\nu_i^1) \, ,
$$
then the dynamical equation
$$
\inn(X)\Omega = \d H \, ,
$$
give rise to the following system of equations
\begin{align*}
\beta_i = \mu_{i+1} \, ,\\
u^0 = \derpar{\Lag}{\mu_0}, \quad  \quad u^i = \derpar{\Lag}{\mu_i} - v^{i-1} \, , \\
v^{k-1} - \derpar{\Lag}{\mu_k} = 0 \, , \\
\xi_1^{i} = \xi^{i} \, , \\
\nu_0^1 = \pounds_g^*\derpar{\Lag}{g} + ad_{\xi_1^0}^*\alpha_0, \quad  \quad \nu_{i+1}^1 = \derpar{\Lag}{\xi^i} - \alpha_i \, ,\\
\alpha_{k-1} - \derpar{\Lag}{\xi^{k-1}} = 0 \, .
\end{align*}
Therefore, the vector field $X$ solution to the dynamical equation is locally given by
\begin{align*}
X &= \mu_{i+1}\derpar{}{\mu_i} + \beta_k\derpar{}{\mu_k} + \derpar{\Lag}{\mu_0}\derpar{}{v^0} + \left( \derpar{\Lag}{\mu_i} - v^{i-1} \right) \derpar{}{v_A^i} \\
&\quad {}+ \xi^0\derpar{}{g} + \xi^{i+1}\derpar{}{\xi^{i}} + \xi_1^{k} \derpar{}{\xi^{k-1}} + \left( L^*_g\derpar{\Lag}{g} + ad_{\xi^0}^*\alpha_0 \right) \derpar{}{\alpha_0}
+ \left( \derpar{\Lag}{\xi^i} - \alpha_i \right) \derpar{}{\alpha_{i+1}}
\end{align*}

Finally, the tangency condition along the submanifold $\W_c$ defined locally by the constraints
\begin{equation}\label{primary}
v^{k-1} - \derpar{\Lag}{\mu_k} = 0 \quad ; \quad
\alpha_{k-1} - \derpar{\Lag}{\xi^{k-1}} = 0 \, ,
\end{equation}
gives the following system of equations for the remaining coefficients $\beta_k$ and $\xi_1^k$
\begin{equation*}
\begin{array}{l}
\derpar{\Lag}{\mu_{k-1}} - v^{k-2} = \mu_{i+1} \derpars{\Lag}{\mu_i}{\mu_k} + \beta_k\derpars{\Lag}{\mu_k}{\mu_k}
+ \xi^0\pounds_g^*\derpars{\Lag}{g}{\mu_k} + \xi^{i+1}\derpars{\Lag}{\xi^i}{\mu_k} + \xi_1^k\derpars{\Lag}{\xi^{k-1}}{\mu_k} \, , \\[10pt]
\derpar{\Lag}{\xi^{k-2}} - \alpha_{k-2} = \mu_{i+1}\derpars{\Lag}{\mu_i}{\xi^{k-1}} + \mu_k\derpars{\Lag}{\mu_k}{\xi^{k-1}}
+ \xi^0\pounds_g^*\derpars{\Lag}{g}{\xi^{k-1}} + \xi^{i+1}\derpars{\Lag}{\xi^i}{\xi^{k-1}} + \xi^k_1 \derpars{\Lag}{\xi^{k-1}}{\xi^{k-1}} \,     .
\end{array}
\end{equation*}

If the matrix
$$
\begin{pmatrix}
\derpars{\Lag}{\mu_k}{\mu_k} & \derpars{\Lag}{\mu_k}{\xi^{k-1}} \\[12pt]
\derpars{\Lag}{\xi^{k-1}}{\mu_k} & \derpars{\Lag}{\xi^{k-1}}{\xi^{k-1}}
\end{pmatrix} \, ,
$$
is regular for every point $p \in (k+1)V^* \times G \times k\g$,
then by a direct computation
the previous equations have a unique solution for
$\mu_{k-1}$ and $\xi^{k}_1$.

\subsection{Constrained problem}
\label{sec:ConstrainedProblem}

\subsubsection{Geometrical setting}

As in Section \ref{sec:UnconstrainedProblem}, let $Q$ be a finite
dimensional smooth manifold modeling the configuration space of a
$k$th-order dynamical system. Now we assume that the dynamics of
the system are constrained. Geometrically, the Lagrangian function
containing the dynamical information of the system is defined at
support on a submanifold of $\Tan^kQ$. Let $j_\N \colon \N
\hookrightarrow \Tan^kQ$ be the constraint submanifold, with
$\codim\N = n$, and $\Lag_\N \in \Cinfty(\N)$ the Lagrangian function
describing the dynamics of the constrained dynamical system.

Let us consider the submanifold
$\overline{\W} = \N \times_{\Tan^{k-1}Q} \Tan^*(\Tan^{k-1}Q)$ of
$\Tan^{k}Q \times_{\Tan^{k-1}Q} \Tan^*(\Tan^{k-1}Q)$ with canonical
embedding $i_{\overline{\W}} \colon \overline{\W} \hookrightarrow \Tan^{k}Q \times_{\Tan^{k-1}Q} \Tan^*(\Tan^{k-1}Q)$
and natural projection $\pr_\N \colon \overline{\W} \to \N$.
If we take $Q = M \times G$, where $M$ is a $m$-dimensional smooth manifold
and $G$ a finite dimensional Lie group, then we have
$\overline{\W} = \N \times_{\Tan^{k-1}M} \Tan^*(\Tan^{k-1}M) \times k\g^*$.

Now, using the results given in Section
\ref{sec:UnconstrainedProblem}, we can define a closed $2$-form in
$\overline{\W}$ as $\overline{\Omega} = i_{\overline{\W}}^*\Omega
\in \df^{2}(\overline{\W})$, where $\Omega \in \df^{2}(\W)$ is the
presymplectic form defined in \eqref{eqn:PresymplecticFormDef}, and
a Hamiltonian function $\overline{H} = i_{\overline{\W}}^* \, \C -
\pr_\N^*\Lag_\N \in \Cinfty(\overline{\W})$, where $\C \in
\Cinfty(\W)$ is the coupling function defined in
\eqref{eqn:CouplingFunctionDef}. With these elements we can state
the dynamical equation for the constrained problem, which is
\begin{equation}\label{eqn:DynEqConstrainedSubmanifold}
\inn(X)\overline{\Omega} = \d\overline{H} \, .
\end{equation}
Since $\N \hookrightarrow \Tan^k(M \times G)$ is an arbitrary
$n$-codimensional submanifold, we do not have a natural set of
coordinates in $\overline{\W}$, and therefore the local study of the
equation can not be done in a general setting. For this reason, we
adopt an ``extrinsic point of view'', that is, we will work in the
bundle $\W$, and then require the solutions to lie in the
submanifold $\overline{\W} \hookrightarrow \W$.

In order to do this, we must construct a Hamiltonian function $H \in \Cinfty(\W)$ using the
Lagrangian function $\Lag_\N \in \Cinfty(\N)$ containing the dynamical information of the system.
Hence, let $\Lag \in \Cinfty(\Tan^k(M \times G))$ be an arbitrary extension of $\Lag_\N$, and let
$H$ be the Hamiltonian function defined in \eqref{eqn:HamiltonianFunctionDef} using this arbitrary
extension of the Lagrangian function $\Lag_\N$.

\subsubsection{Dynamical equation}

The extrinsic dynamical equation for a constrained dynamical system is
\begin{equation}\label{eqn:DynEqConstrainedExtrinsic}
\inn(X)\Omega - \d H \in \operatorname{ann}(\Tan\overline{\W}) \ , \quad
\mbox{for } X \in \vf(\W) \mbox{ tangent to } \overline{\W} \, .
\end{equation}
where $\operatorname{ann}(D)$ denotes the annihilator of a distribution
$D \subset \Tan\W$. Observe that this equation is clearly equivalent to \eqref{eqn:DynEqConstrainedSubmanifold}.

Then, following \cite{art:Gotay_Nester_Hinds78} we have
\begin{proposition}\label{prop:CompatibilitySubmanifoldConstrained}
A solution to the equation \eqref{eqn:DynEqConstrainedExtrinsic} exists only on the points of the
submanifold $\W_c \hookrightarrow \W$ defined by
$$
\W_c = \left\{ p \in \W \colon (\inn(Y)\d H)(p) \in \operatorname{ann}(\Tan_p\overline{\W})
\, , \ \forall \, Y \in \ker\Omega \right\}.
$$
\end{proposition}

In natural coordinates, let $\Phi^a \in \Cinfty(\Tan^k(M \times G))$, $1 \leqslant a \leqslant n$, be
local functions defining the submanifold $\N \hookrightarrow \Tan^k(M \times G)$, that is,
$$
\N = \left\{ p \in \Tan^k(M \times G) \colon \Phi^a(p) = 0 \, , \, 1 \leqslant a \leqslant n \right\}.
$$
In an abuse of notation, we also denote by $\Phi^a$ the pull-back
of the constraint functions to $\W$.
Then, the annihilator of $\Tan\overline{\W}$ is locally given by
$$
\operatorname{ann}(\Tan\overline{\W}) = \left\langle \d\Phi^a \right\rangle \, .
$$
Therefore, the equation defining the submanifold $\W_c$ may be written locally as
$$
\inn(Y)\d H = \lambda_a\d\Phi^a \, , \ \forall \, Y \in \ker\Omega \, ,
$$
where $\lambda_a$, $1 \leqslant a \leqslant n$ are the Lagrange
multipliers. Then, bearing in mind the local expression
\eqref{eqn:DiffHamiltonianLocal} of $\d H$ and
\eqref{eqn:KernelOmegaLocal} of $\ker\Omega$, the equations defining
locally the submanifold $\W_c$ are
$$
p_A^i - \derpar{L}{q_k^A} + \lambda_a\derpar{\Phi^a}{q_k^A} = 0 \quad ; \quad
\alpha_{k-1} - \derpar{L}{\xi^{k-1}} + \lambda_a\derpar{\Phi^a}{\xi^{k-1}} = 0 \quad ; \quad
\Phi^a = 0 \, .
$$

Now, let us compute the local expression of equation \eqref{eqn:DynEqConstrainedExtrinsic}.
If we assume that $\N$ is determined by the vanishing of the $n$ functions $\Phi^a$, then equation
\eqref{eqn:DynEqConstrainedExtrinsic} may be rewritten as
$$
\inn(X)\Omega - \d H = \lambda_a\d\Phi^a \, ,
$$
where $\lambda_a$ are Lagrange multipliers to be determined. Then, bearing in mind the local
expression \eqref{eqn:DiffHamiltonianLocal} of $\d H$ and \eqref{eqn:OmegaLocal} of $\Omega$,
taking a generic vector field locally given by \eqref{eqn:GenericVectorField} we obtain the following
system of equations
\begin{align}
F_i^A = q_{i+1}^A ,\,  \\
G_A^{0} = \derpar{\Lag}{q_0^A} - \lambda_a\derpar{\Phi^a}{q_0^A} \quad , \quad
G_A^i = \derpar{\Lag}{q_i^A} - \lambda_a\derpar{\Phi^a}{q_i^A} - p_A^{i-1} \, , \\
p_A^{k-1} - \derpar{\Lag}{q_k^A} + \lambda_a\derpar{\Phi^a}{q_k^A} = 0 \, , \label{eqn:DynEqLocalHOMomentaManifoldConstrained} \\
\xi_1^{i} = \xi^{i} \, , \\
\nu_0^1 = \pounds_g^*\left( \derpar{\Lag}{g} - \lambda_a\derpar{\Phi^a}{g} \right) + ad_{\xi_1^0}^*\alpha_0 \quad , \quad
\nu_{i+1}^1 = \derpar{\Lag}{\xi^i} - \lambda_a\derpar{\Phi^a}{\xi^i} - \alpha_i \, , \\
\alpha_{k-1} - \derpar{\Lag}{\xi^{k-1}} + \lambda_a\derpar{\Phi^a}{\xi^{k-1}} = 0 \label{eqn:DynEqLocalHOMomentaGroupConstrained} \, , \\
\Phi^a(q_i^A,q_k^A,g,\bfxi,\xi^{k-1}) = 0 \, . \label{eqn:LocalConstraintsSubmanifold}
\end{align}
Therefore, the vector field $X$ solution to equation \eqref{eqn:DynEqConstrainedExtrinsic}
is locally given by
\begin{align*}
X &= q_{i+1}^A\derpar{}{q_i^A} + F_k^A\derpar{}{q_k^A}
+ \left(\derpar{\Lag}{q_0^A} - \lambda_a\derpar{\Phi^a}{q_0^A} \right)\derpar{}{p_A^0}
+ \left( \derpar{\Lag}{q_i^A} - \lambda_a\derpar{\Phi^a}{q_i^A} - p_A^{i-1} \right) \derpar{}{p_A^i} \\
&\quad {}+ \xi^0\derpar{}{g} + \xi^{i+1}\derpar{}{\xi^{i}} + \xi_1^{k} \derpar{}{\xi^{k-1}}
+ \left( \pounds_g^*\left( \derpar{\Lag}{g} - \lambda_a\derpar{\Phi^a}{g} \right) + ad_{\xi^0}^*\alpha_0 \right) \derpar{}{\alpha_0} \\
&\quad {}+ \left( \derpar{\Lag}{\xi^i} - \lambda_a\derpar{\Phi^a}{\xi^i} - \alpha_i \right) \derpar{}{\alpha_{i+1}}
\end{align*}
Observe that equations \eqref{eqn:DynEqLocalHOMomentaManifoldConstrained},
\eqref{eqn:DynEqLocalHOMomentaGroupConstrained} and \eqref{eqn:LocalConstraintsSubmanifold}
do not involve coefficient functions of the vector field $X$: they are pointwise algebraic relations, stating
that the vector field $X$ exists with support on a submanifold defined locally by these equations.
Hence, we recover locally the result stated in Proposition \ref{prop:CompatibilitySubmanifoldConstrained}.

The coefficients $F_k^A$ and $\xi_1^k$ of the vector field and the Lagrange multipliers $\lambda_a$
remain undetermined. Nevertheless, from Proposition \ref{prop:CompatibilitySubmanifoldConstrained}
we know that the vector field $X$ exists only at support on the submanifold $\W_c$. Hence, we must require
the vector field $X$ to be tangent to $\W_c$, that is, we must impose $\restric{\Lie(X)\zeta}{\W_c} = 0$ for
every constraint function $\zeta$ defining $\W_c$. Then, taking into account that $\W_c$ is locally defined by
equations \eqref{eqn:DynEqLocalHOMomentaManifoldConstrained},
\eqref{eqn:DynEqLocalHOMomentaGroupConstrained} and \eqref{eqn:LocalConstraintsSubmanifold},
the tangency condition for $X$ along $\W_c$ gives the following equations
\begin{align}
\derpar{\Lag}{q_{k-1}^A} - p_A^{k-2}
&= q_{i+1}^B \left( \derpars{\Lag}{q_i^B}{q_k^A} - \lambda_a\derpars{\Phi^a}{q_i^B}{q_k^A} \right)
+ F_k^B \left( \derpars{\Lag}{q_k^B}{q_k^A} - \lambda_a\derpars{\Phi^a}{q_k^B}{q_k^A} \right) \nonumber \\
&\quad {} + \xi^0 \pounds_g^*\left( \derpars{\Lag}{g}{q_k^A} - \lambda_a\derpars{\Phi^a}{g}{q_k^A} \right)
+ \xi^{i+1} \left( \derpars{\Lag}{\xi^i}{q_k^A} - \lambda_a \derpars{\Phi^a}{\xi^{i}}{q_k^A} \right) \nonumber \\
&\quad {}+ \xi^k_1\left( \derpars{\Lag}{\xi^{k-1}}{q_k^A} - \lambda_a\derpars{\Phi^a}{\xi^{k-1}}{q_k^A} \right)
+ \lambda_a\derpar{\Phi^a}{q_{k-1}^A} \nonumber \\
\derpar{\Lag}{\xi^{k-2}} - \alpha_{k-2}
&= q_{i+1}^A \left( \derpars{\Lag}{q_i^A}{\xi^{k-1}} - \lambda_a\derpars{\Phi^a}{q_i^A}{\xi^{k-1}} \right)
+ F_k^A \left( \derpars{\Lag}{q_k^A}{\xi^{k-1}} - \lambda_a\derpars{\Phi^a}{q_k^A}{\xi^{k-1}} \right) \label{eqn:TangencyConditionConstrained} \\
&\quad {} + \xi^0 \pounds_g^*\left( \derpars{\Lag}{g}{\xi^{k-1}} - \lambda_a\derpars{\Phi^a}{g}{\xi^{k-1}} \right)
+ \xi^{i+1} \left( \derpars{\Lag}{\xi^i}{\xi^{k-1}} - \lambda_a \derpars{\Phi^a}{\xi^{i}}{\xi^{k-1}} \right) \nonumber \\
&\quad {}+ \xi^k_1\left( \derpars{\Lag}{\xi^{k-1}}{\xi^{k-1}} - \lambda_a\derpars{\Phi^a}{\xi^{k-1}}{\xi^{k-1}} \right)
+ \lambda_a\derpar{\Phi^a}{\xi^{k-2}} \nonumber \\
q_{i+1}^A\derpar{\Phi^a}{q_i^A} + F_k^A&\derpar{\Phi^a}{q_k^A} + \xi^0\pounds_g^*\derpar{\Phi^a}{g}
+ \xi^{i+1}\derpar{\Phi^a}{\xi^i} + \xi_1^k\derpar{\Phi^a}{\xi^{k-1}} = 0. \nonumber
\end{align}

If we denote by $\Omega_{\W_c}$ the pullback of the presymplectic
2-form $\Omega$ to $\W_c,$ then we deduce the following theorem.

\begin{theorem}
$(\W_{c},\Omega_{\W_c})$ is a symplectic manifold if and only if
$$
\begin{pmatrix}
\derpars{\Lag}{q_k^B}{q_k^A} - \lambda_a\derpars{\Phi^a}{q_k^B}{q_k^A} & \ & \derpars{\Lag}{q_k^A}{\xi^{k-1}} - \lambda_a\derpars{\Phi^a}{q_k^A}{\xi^{k-1}} & \ & \left(\derpar{\Phi^a}{q_{k}^A}\right)^T \\[15pt]
\derpars{\Lag}{\xi^{k-1}}{q_k^A} - \lambda_a\derpars{\Phi^a}{\xi^{k-1}}{q_k^A} & \ & \derpars{\Lag}{\xi^{k-1}}{\xi^{k-1}} - \lambda_a\derpars{\Phi^a}{\xi^{k-1}}{\xi^{k-1}} & \ & \left(\derpar{\Phi^a}{\xi^{k-1}}\right)^T \\[15pt]
\derpar{\Phi^a}{q_{k}^A} & \ & \derpar{\Phi^a}{\xi^{k-1}} & \ & \mathbf{0}
\end{pmatrix}
$$
is nondegenerate along $\W_{c}.$
\end{theorem}
\begin{proof}
The proof of this theorem is a straightforward computation using
Theorem 4.1 in \cite{art:Colombo_Martin_Zuccalli10} and Theorem 3.3
in \cite{art:Colombo_Martin11}.
\end{proof}

Now, let $\gamma \colon \R \to \W$ be an integral curve of $X$
locally given by \eqref{eqn:GenericCurve}. Then the condition $X
\circ \gamma = \dot{\gamma}$ gives the following system of
differential equations for the component functions of $\gamma$
\begin{align}
\dot{q}_1^A = q_{i+1}^A, \,  \\
\dot{p}_A^0 = \derpar{\Lag}{q_0^A} - \lambda_a\derpar{\Phi^a}{q_0^A} \quad , \quad
\dot{p}_A^i = \derpar{\Lag}{q_i^A} - \lambda_a\derpar{\Phi^a}{q_i^A} - p_A^{i-1}, \,  \label{eqn:DynEqCurveLocalManifoldConstrained} \\
\dot{g} = g\xi^0 \quad , \quad \dot{\xi}^{i-1} = \xi^{i} \, , \\
\dot{\alpha}_0 = \pounds_g^*\left( \derpar{\Lag}{g} - \lambda_a\derpar{\Phi^a}{g} \right) + ad_{\xi_1^0}^*\alpha_0 \quad , \quad
\dot{\alpha}_{i+1} = \derpar{\Lag}{\xi^i} - \lambda_a\derpar{\Phi^a}{\xi^i} - \alpha_i \, , \label{eqn:DynEqCurveLocalGroupConstrained}
\end{align}
in addition to equations
\eqref{eqn:DynEqLocalHOMomentaManifoldConstrained},
\eqref{eqn:DynEqLocalHOMomentaGroupConstrained} and
\eqref{eqn:LocalConstraintsSubmanifold}. Now, using equations
\eqref{eqn:DynEqLocalHOMomentaManifoldConstrained} in combination
with \eqref{eqn:DynEqCurveLocalManifoldConstrained} we obtain the
$k$th order constrained Euler-Lagrange equations
\begin{equation}\label{eqn:EulerLagrangeConstrained}
\restric{\sum_{i=0}^{k}(-1)^{i}\frac{d^i}{dt^i}\left( \derpar{\Lag}{q_i^A} - \lambda_a\derpar{\Phi^a}{q_i^A} \right)}{\gamma} = 0 \, .
\end{equation}
On the other hand, using equations \eqref{eqn:DynEqLocalHOMomentaGroupConstrained} in combination
with \eqref{eqn:DynEqCurveLocalManifoldConstrained} we obtain the $k$th order trivialized constrained
Euler-Lagrange equation
\begin{equation}\label{eqn:TrivializedEulerLagrangeConstrained}
\restric{\left( \frac{d}{dt} - ad_{\xi^0}^* \right)\sum_{i=0}^{k-1}(-1)^i\frac{d^i}{dt^i}\left( \derpar{\Lag}{\xi^i} - \lambda_a\derpar{\Phi^a}{\xi^i} \right)}{\gamma}
= \restric{\pounds_g^*\left( \derpar{\Lag}{g} - \lambda_a\derpar{\Phi^a}{g} \right)}{\gamma} \, .
\end{equation}
Therefore, a dynamical trajectory $\gamma \colon \R \to \W$ of the system must satisfy the equations
\eqref{eqn:EulerLagrangeConstrained} and \eqref{eqn:TrivializedEulerLagrangeConstrained}, in addition
to $\Phi^a(q_i^A(t),q_k^A(t),g(t),\xi^i(t)) = 0$.

Finally, if both the extended Lagrangian function $\Lag \in
\Cinfty(\Tan^kM \times G \times k\g)$ and the constraint functions
$\Phi^a \in \Cinfty(\Tan^{k}M \times G \times k\g)$ are
left-invariant, then we can define the reduced Lagrangian function
$\ell \in \Cinfty(\Tan^kM \times k\g)$ and the reduced constraint
functions $\phi^a \in \Cinfty(\Tan^kM \times k\g)$ as
$$
\ell(q_i^A,q_k^A,\xi^i) = \Lag(q_i^A,q_k^A,e,\xi^i) \quad , \quad
\phi^a(q_i^A,q_k^A,\xi^i) = \Phi^a(q_i^A,q_k^A,e,\xi^i) \, ,
$$
and then equations \eqref{eqn:TrivializedEulerLagrangeConstrained} become
$$
\restric{\left( \frac{d}{dt} - ad_{\xi^0}^* \right)\sum_{i=0}^{k-1}(-1)^i\frac{d^i}{dt^i}\left( \derpar{\ell}{\xi^i} - \lambda_a\derpar{\phi^a}{\xi^i} \right)}{\gamma} = 0 \, .
$$
Note that equations \eqref{eqn:EulerLagrangeConstrained} remain the
same, just replacing $\Lag$ by $\ell$ and $\Phi^a$ by $\phi^a$.

\section{Application to optimal control of underactuated mechanical systems}
\label{sec:OptimalControlUnderactuated}

In this section we study optimal control problems for
underactuated mechanical systems (or superarticulated mechanical
systems following the terminology in \cite{art:Baillieul99}). The presence of underactuated
mechanical systems is ubiquitous in engineering applications as a
result, for instance, of design choices motivated by the search of
less cost devices or as a result of a failure regime in fully
actuated mechanical systems. The underactuated systems include
spacecraft, underwater vehicles, mobile robots, helicopters, wheeled
vehicles, mobile robots, underactuated manipulators, etc.

Let $U \subseteq \R^r$ be the control manifold where $u(t) \in U$ is the control parameter.
We assume that all the control systems are controllable, that is, for any two points $q_0$
and $q_T$ in the configuration space $Q$, there exists an admissible control $u(t)$ defined
on some interval $[0,T] \subset \R$ such that the system with initial condition $q_0$ reaches
the point $q_T$ in time $T$ (see \cite{book:Bullo_Lewis05} for details).

Let us consider that the configuration space $Q$ of the system is a trivial principal bundle,
that is, $Q = M \times G$, where $M$ is an $m$-dimensional smooth manifold and $G$
a finite dimensional Lie group. Let $\Lag \in \Cinfty(\Tan M \times \g)$ be a left-trivialized
Lagrangian function, where $\g$ is the Lie algebra of $G$.

The Euler-Lagrange equations with controls are
\begin{equation}\label{eqn:EulerLagrangeControls}
\begin{array}{l}
\displaystyle \frac{d}{dt} \left( \derpar{\Lag}{\dot{q}^A} \right) - \derpar{\Lag}{q^A} = u_a\mu_A^a(q) \, , \quad (1 \leqslant A \leqslant m) \\[10pt]
\displaystyle \frac{d}{dt} \left( \derpar{\Lag}{\xi} \right) - ad_\xi^*\left( \derpar{\Lag}{\xi} \right) = u_a\eta^a(q) \, ,
\end{array}
\end{equation}
where $\mathcal{B}^a = \{ (\mu^a,\eta^a) \}$, $\mu^a(q) \in \Tan^*_qM$, $\eta^a(q) \in \g^*$, $a = 1,\ldots,r$,
is a set of independent  sections of the bundle $\pi \colon \Tan^*M \times \g^* \to M$, and $u_a$ are admissible controls.

We complete $\mathcal{B}^a$ to a basis $\{ \mathcal{B}^a,\mathcal{B}^\alpha \}$
of $\Gamma(\pi)$, and let us consider its dual basis $\{ \mathcal{B}_a,\mathcal{B}_\alpha \}$,
that is, a basis of $\Gamma(\tau)$, where $\tau \colon \Tan M \times \g \to M$. Observe that
$\Gamma(\tau) = \vf(M) \times \Cinfty(M,\g)$ (see \cite{art:DeLeon_Marrero_Martinez05} for details).
This basis induces coordinates $(q^A,\dot{q}^A,\xi^a,\xi^\alpha)$ on $\Tan M \times \g$.

If we denote $\mathcal{B}_a = \{ (X_a, \Xi_a) \} \subset \Gamma(\tau)$ and
$\mathcal{B}_\alpha = \{ X_\alpha, \Xi_\alpha \} \subset \Gamma(\tau)$,
where $X_a,X_\alpha \in \vf(M)$ are locally given by $X_a(q) = X_a^A(q)\restric{\derpar{}{q^A}}{q}$,
$X_\alpha(q) = X_\alpha^A(q)\restric{\derpar{}{q^A}}{q}$, and $\Xi_a(q), \Xi_\alpha(q) \in \g$, with $q \in M$,
then equations \eqref{eqn:EulerLagrangeControls} can be rewritten as
\begin{equation}\label{eqn:EulerLagrangeControlsRewritten}
\begin{array}{l}
\displaystyle \left( \frac{d}{dt}\left( \derpar{\Lag}{\dot{q}^A} \right) - \derpar{\Lag}{q^A} \right)X_a^A(q)
+ \left( \frac{d}{dt}\left( \derpar{\Lag}{\xi} \right) - ad_\xi^*\derpar{\Lag}{\xi} \right)\Xi_a(q) = u_a \, , \\[10pt]
\displaystyle \left( \frac{d}{dt}\left( \derpar{\Lag}{\dot{q}^A} \right) - \derpar{\Lag}{q^A} \right)X_\alpha^A(q)
+ \left( \frac{d}{dt}\left( \derpar{\Lag}{\xi} \right) - ad_\xi^*\derpar{\Lag}{\xi} \right)\Xi_\alpha(q) = 0 \, .
\end{array}
\end{equation}

\paragraph{Optimal control problem.}

The optimal control problem consists in finding a trajectory $(q(t),\dot{q}(t),\xi(t),u(t))$ of the state
variables and control inputs solving equation \eqref{eqn:EulerLagrangeControlsRewritten}
given initial and final conditions $(q(0),\dot{q}(0),\xi(0))$ and $(q(T),\dot{q}(T),\xi(T))$, respectively,
and minimizing the following functional
$$
\mathcal{A}(q,\dot{q},\xi,u) = \int_{0}^{T} C(q(t),\dot{q}(t),\xi(t),u(t))dt \, ,
$$
where $C \colon (\Tan M \times \g) \times U \to \R$ is a cost function.

Following \cite{book:Bloch03}, to solve this optimal control problem is equivalent to solve the following
second-order variational problem with second-order constraints
\begin{align*}
&\textnormal{min } \widetilde{\Lag}(q^A,\dot{q}^A,\ddot{q}^A,\xi^i,\dot{\xi}^i) \\
&\textnormal{subject to } \Phi^\alpha(q^A,\dot{q}^A,\ddot{q}^A,\xi^i,\dot{\xi}^i) \, , \, \alpha = 1,\ldots,m
\end{align*}
where $\widetilde{\Lag}, \Phi^\alpha \in \Cinfty(\Tan^2M \times 2\g)$ are given by
$$
\widetilde{\Lag}(q^A,\dot{q}^{A},\ddot{q}^{A},\xi^i,\dot{\xi}^i)
=C\left(q^{A},\dot{q}^{A},\xi^i,F_{a}(q^{A},\dot{q}^{A},\ddot{q}^{A},\xi^{i},\dot{\xi}^{i})\right) \, ,
$$
where $C$ is the cost function and
$$
F_{a}(q^{A},\dot{q}^{A},\ddot{q}^{A},\xi^{i},\dot{\xi}^{i}) =
\left(\frac{d}{dt}\left(\derpar{\Lag}{\dot{q}^A}\right)-\derpar{\Lag}{q^A}\right)X_a^{A}(q)
+\left(\frac{d}{dt}\left(\derpar{\Lag}{\xi}\right)-\left(ad_{\xi}^{*}\derpar{\Lag}{\xi}\right)\right)\Xi_a(q) \, .
$$
The Lagrangian $\widetilde{\Lag}$ is subjected to the second-order
constraints:
$$
\Phi^{\alpha}(q^{A},\dot{q}^{A},\ddot{q}^{A},\xi^i,\dot{\xi}^i)=
\left(\frac{d}{dt}\left(\derpar{\Lag}{\dot{q}^{A}}\right) - \derpar{\Lag}{q^{A}}\right)X_{\alpha}^{A}(q)
+\left(\frac{d}{dt}\left(\derpar{\Lag}{\xi}\right)-\left(ad_{\xi}^{*}\derpar{\Lag}{\xi}\right)\right)\Xi_{\alpha}(q) \, .
$$

\subsection{Optimal control of an underactuated vehicle}

Consider a rigid body moving in special Euclidean group of the plane
$SE(2)$ with a thruster to adjust its pose. The configuration of
this system is determined by a tuple $(x, y, \theta, \gamma)$, where
$(x, y)$ is the position of the center of mass, $\theta$ is the
orientation of the blimp with respect to a fixed basis, and $\gamma$
the orientation of the thrust with respect to a body basis.
Therefore, the configuration manifold is  $Q = SE(2)\times
\mathbb{S}^1$ (see \cite{book:Bullo_Lewis05} and references
therein), where $(x, y, \theta)$ are the local coordinates of
$SE(2)$ and $\gamma$ is the local coordinate of $\mathbb{S}^1$.

The Lagrangian of this system is given by its kinetic energy
$$
\Lag(x, y, \theta, \gamma, \dot{x}, \dot{y}, \dot{\theta}, \dot{\gamma})=\frac{1}{2}m(\dot{x}^2+\dot{y}^2)+\frac{1}{2}J_1\dot{\theta}^2+\frac{1}{2}J_2( \dot{\theta} + \dot{\gamma})^2 \, ,
$$
and the input forces are
$$
F^1 = \cos(\theta+\gamma)\,\d x+\sin (\theta+\gamma)\, \d y - p\sin \gamma \d\theta \quad ; \quad
F^2 = \d\gamma,
$$
where the control forces that we consider are applied to a point on
the body with distance $p > 0$ from the center of mass ($m$ is the
mass of the rigid body), along the body $x$-axis. Note this system
is an example of underactuated mechanical system when the
configuration space is a trivial principal bundle.

The system is invariant under the left multiplication of the Lie
group $G=SE(2)$:
\[
\begin{array}{rcl}
\Phi\colon SE(2)\times SE(2)\times \mathbb{S}^{1}&\longrightarrow& SE(2)\times \mathbb{S}^{1}\\
((a, b, \alpha), (x, y, \theta, \gamma))&\longmapsto&
(x\cos\alpha-y\sin \alpha +a, x\sin\alpha + y\cos\alpha + b, \theta+\alpha, \gamma).
\end{array}
\]
A basis of the Lie algebra $\mathfrak{se}(2) \simeq \R ^3$ of $SE(2)$
is given by
$$
e_1=
\begin{pmatrix}
0 & -1 & 0\\
1 & 0 & 0\\
0 & 0 & 0
\end{pmatrix},\quad
e_2=
\begin{pmatrix}
0 & 0 & 1\\
0 & 0 & 0\\
0 & 0 & 0
\end{pmatrix},\quad
e_3=
\begin{pmatrix}
0 & 0 & 0\\
0 & 0 & 1\\
0 & 0 & 0
\end{pmatrix},
$$
from we have that
$$
[e_1,e_2]=e_3,\quad [e_1,e_3]=-e_2,\quad [e_2,e_3]=0.
$$
Thus, we can write down the structure constants as
$$
\mathcal{C}_{31}^{2}=\mathcal{C}_{23}^{1}=-1, \mathcal{C}_{13}^2=\mathcal{C}_{32}^{1}=1 \, ,
$$
and all others vanish. An element $\xi \in \mathfrak{se}(2)$ is of the
form $ \xi =\xi_1\, e_1+\xi_2\, e_2+\xi_3\, e_3;$ therefore the
reduced Lagrangian $\ell\colon\Tan\mathbb{S}^1\times\mathfrak{se}(2)\to\R$
is given by
$$
\ell(\gamma,\dot{\gamma},\xi)=\frac{1}{2}m(\xi_1^2+\xi_2^2)+\frac{J_1+J_2}{2}\xi_3^2+J_2\xi_3\dot{\gamma}+\frac{J_2}{2}\dot{\gamma}^2 \, .
$$
Then the reduced Euler-Lagrange equations with controls are given by
(see, for example, \cite{art:Colombo_Jimenez_Martin12} and
\cite{art:Colombo_Jimenez_Martin12-2})
\begin{eqnarray*}
m \dot{\xi_1}&=&u_1\cos{\gamma} \, ,\\
m\dot{\xi_2}+(J_1+J_2)\xi_1\xi_3+J_2\xi_1\dot{\gamma}-m\xi_1\xi_3&=&u_1\sin{\gamma} \, ,\\
(J_1+J_2)\dot{\xi}_3+J_2\ddot{\gamma}-m\xi_2(\xi_1+\xi_3)&=&-u_1p\sin{\gamma} \, ,\\
J_2(\dot{\xi}_3+\ddot{\gamma})&=&u_2 \, .
\end{eqnarray*}
On the other hand, choosing the adapted basis
$\{\mathcal{B}_a,\mathcal{B}_{\alpha}\}$ the modified equations of
motion \eqref{eqn:EulerLagrangeControlsRewritten} read in this case as
\begin{eqnarray*}
m(\cos{\gamma}\dot{\xi}_1+\sin{\gamma}(\dot{\xi}_2-\xi_1\xi_3))+(J_1+J_2)\xi_1\xi_3\sin\gamma+J_2\xi_1\dot{\gamma}\sin\gamma&=&u_1 \, ,\\
m(\cos{\gamma}(\dot{\xi}_2-\xi_1\xi_3)-\sin{\gamma}\dot{\xi}_1)+\xi_1\xi_3(J_1+J_2)\cos\gamma+J_2\xi_1\dot{\gamma}\cos{\gamma}&=&0 \, ,\\
\frac{J_1+J_2}{p}(\dot{\xi}_3+p\xi_1\xi_3)+\frac{J_2}{p}(\ddot{\gamma}+p\xi_1\dot{\gamma})+m\left(\dot{\xi}_2-\xi_1\xi_3-\frac{\xi_2\xi_1+\xi_3\xi_2}{p}\right)&=&0 \, ,\\
J_2(\dot{\xi}_3+\ddot{\gamma})&=&u_2 \, .
\end{eqnarray*}
Now, we can study the optimal control problem that consists, as
mentioned before, on finding a trajectory of state variables and
control inputs satisfying the previous equations from given initial
and final conditions $(\gamma(0),\dot{\gamma}(0),\xi(0)$),
$(\gamma(T),\dot{\gamma}(T),\xi(T))$ respectively,  and extremizing
the cost functional
$$
\int^T_0 (\rho_1u_1^2+\rho_2u_2^2)\; \d t \, ,
$$
where $\rho_1$ and $\rho_2$ are non-zero constants.

The related optimal control problem is equivalent to the
second-order Lagrangian problem with second-order constraints
defined as follows. Extremize
$$
\widetilde{\mathcal{A}}=\int_{0}^{T}\widetilde{\Lag}(\xi,\dot{\xi},\gamma,\dot{\gamma},\ddot{\gamma})\d t,
$$
subject to second-order constraints given by the functions
\begin{subequations}\label{Phid}
\begin{align}
\Phi^{1}&=m(\cos{\gamma}(\dot{\xi}_2-\xi_1\xi_3)-\sin{\gamma}\dot{\xi}_1)+\xi_1\xi_3(J_1+J_2)\cos{\gamma}+J_2\xi_1\dot{\gamma}\cos{\gamma} \, ,\label{Phida}\\
\Phi^{2}&=\frac{J_1+J_2}{p}(\dot{\xi}_3+p\xi_1\xi_3)+\frac{J_2}{p}(\ddot{\gamma}+p\xi_1\dot{\gamma})+m\left(\dot{\xi}_2-\xi_1\xi_3-\frac{\xi_2\xi_1+\xi_3\xi_2}{p}\right) \, .\label{Phib}
\end{align}
\end{subequations}
Here, $\widetilde{\Lag} \colon \Tan^{2}\mathbb{S}^{1}\times 2\mathfrak{se}(2)\to\R$
is defined by
\begin{align}\label{tildeL}
\widetilde{\Lag}(\gamma,\dot{\gamma},\ddot{\gamma},\xi,\dot{\xi}) &=
\rho_1\left(m(\cos{\gamma}\dot{\xi}_1+\sin{\gamma}(\dot{\xi}_2-\xi_1\xi_3))+(J_1+J_2)\xi_1\xi_3\sin\gamma+J_2\xi_1\dot{\gamma}\sin\gamma\right)^{2} \\
& \quad {} +\rho_2J_2^2(\dot{\xi}_3+\ddot{\gamma})^2 \, , \nonumber
\end{align}
which basically is the cost function $C=\rho_1u_1^2+\rho_2u_2^2$ in
terms of the new variables.

The submanifold $\mathcal{M}$ of $\Tan^{2}\mathbb{S}^{1}\times SE(2)\times 2\mathfrak{se}(2)$
is given by the constraints
\begin{align*}
&\ddot{\gamma} = -\frac{mp}{J_2}\left(\dot{\xi}_{2}-\xi_{1}\xi_{3}-\xi_{2}\frac{\xi_{1}+\xi_{3}}{p}\right)-\frac{J_1+J_2}{J_2}(\dot{\xi}_{3}+p\xi_1\xi_3)-p\xi_1\dot{\gamma} \, , \\
&\dot{\xi}_{1} = \frac{1}{\tan\gamma}\left(\frac{J_1+J_2}{m}\xi_1\xi_3+\frac{J_2}{m}\xi_1\dot{\gamma}+\dot{\xi}_{2}-\xi_{1}\xi_{3}\right).
\end{align*}

We consider the submanifold $\overline{\W}=\mathcal{M}\times
\Tan^{*}(\Tan\mathbb{S}^{1})\times 2\mathfrak{se}(2)^{*}$ with induced
coordinates
$$
(\gamma,\dot{\gamma},g,\xi_1,\xi_2,\xi_3,\dot{\xi}_{2},\dot{\xi}_{3},\eta_1,\eta_2,
p_1, p_2, p_3, \tilde{p}_1,\tilde{p}_2,\tilde{p}_3) \, ,
$$
and the restriction $\widetilde{\Lag}_{\mathcal{M}}$ of $\widetilde{\Lag}$
given by
\begin{align*}
\widetilde{\Lag}_{\mathcal{M}}=& \, \rho_{1}\left[\frac{m\cos\gamma}{\tan\gamma}\left(\frac{J_1+J_2}{m}\xi_1\xi_3+\frac{J_2}{m}\xi_{1}\dot{\gamma}+\dot{\xi}_{2}-\xi_1\xi_3\right)+(J_1+J_2)\xi_1\xi_3\sin\gamma+J_2\xi_1\dot{\gamma}\sin\gamma\right.\\
& \left.+\sin\gamma(\dot{\xi}_{2}-\xi_1\xi_3)\right]^{2}+\rho_2J_{2}^{2}\left[\dot{\xi}_{3}-\frac{mp}{J_2}
\left(\dot{\xi}_{2}-\xi_{1}\xi_{3}-\xi_{2}\frac{\xi_1+\xi_3}{p}\right)
-\frac{J_1+J_2}{J_2}(\dot{\xi}_{3}+p\xi_1\xi_3)-p\xi_{1}\dot{\gamma}\right]^{2}.
\end{align*}

Observe that we use the intrinsic formulation in the submanifold
$\mathcal{M}$ because the constraints enable us to write the variables
$\ddot{\gamma}$ and $\dot{\xi}_{1}$ in terms of the others,
and thus it is easy to determine a subset of
intrinsic coordinates.

Then, we can write the equations of motion of the optimal control
problem for this underactuated system. For simplicity, we consider
the particular case $J_1=J_2=1$ and $m=p=1$ then the equations of
motion of the optimal control system are:
\begin{align*}
\dot{\xi}_{i} &= \frac{d}{dt}\xi_{i} \quad , \quad \dot{\gamma}=\frac{d}{dt}\gamma \quad , \quad
\ddot{\gamma}=\frac{d}{dt}\dot{\gamma} \ ,\quad i=1,2,3.\\
\dot{\eta}_{1} &= 2\rho_{1}\left(\left(\frac{\cos\gamma}{\tan\gamma} + \sin\gamma\right)\cdot\mathcal{A}\right)\left(\cos\gamma-\frac{\sin\gamma}{\tan\gamma}-\frac{1}{\cos^{2}\gamma+\tan^{2}\gamma}\right)+\lambda_{1}\sin\gamma\mathcal{A}\left(1+\frac{1}{\tan\gamma}\right) \, ,\\
\dot{\eta}_{2} &= 2\xi_{1}\rho_{1}\left(\left(\frac{\cos\gamma}{\tan\gamma}+\sin\gamma\right)^{2}\cdot\mathcal{A}\right)-\xi_{1}(\lambda_{1}\cos\gamma+\lambda_{2}-2\mathcal{B}\rho_{2})-\eta_{1} \, ,\\
\dot{\tilde{p}}_{1} &= 2\rho_{1}\left(\left(\frac{\cos\gamma}{\tan\gamma}+\sin\gamma\right)^{2}\cdot\mathcal{A}\right)(\mathcal{B}-\lambda_{2})(\xi_{3}+\dot{\gamma}-\xi_{2})-\lambda_{1}\cos\gamma(\dot{\gamma}+\xi_{3}) \, ,\\
\dot{\tilde{p}}_{2} &= (\lambda_{2}-2\mathcal{B}\rho_{2})(\xi_{1}+\xi_{3}) \, ,\\
\dot{\tilde{p}}_{3} &= 2\xi_{1}\rho_{1}\left(\left(\frac{\cos\gamma}{\tan\gamma}+\sin\gamma\right)^{2}\cdot\mathcal{A}\right)-\lambda_{1}\cos\gamma\xi_{1}-(\xi_{1}-\xi_{2})(\lambda_{2}+2\rho_{2}\mathcal{B})\\
\dot{p}_{i} &= ad_{\xi}^{*}p_{i} \, ,\quad i=1,2,3.
\end{align*}
where
$$
\xi=(\xi_{1},\xi_{2},\xi_{3}) \quad ; \quad
\mathcal{A}=\xi_{1}\xi_{3}+\xi_{1}\dot{\gamma}+\dot{\xi}_{2} \quad ; \quad
\mathcal{B}=\dot{\xi}_{3}+\xi_{1}\xi_{3}+\dot{\xi}_{2}-\xi_{2}\xi_{1}-\xi_{2}\xi_{3}+\xi_{1}\dot{\gamma}
$$
and the coadjoint operator is just the cross product,
$ad^{*}_{\xi}p=\xi\times p$ using the identification of
$\mathfrak{se}(2)$ with $\R^{3}.$ One can check that, using
different techniques, these equations are the equations of control
in the classical literature as \cite{book:Bullo_Lewis05}. Also one
can compare these equations with the equations obtained using
variational tools in \cite{art:Colombo_Jimenez_Martin12} and
\cite{art:Colombo_Jimenez_Martin12-2}.

In all cases we additionally have the reconstruction equation
$$
\dot{g} (t) = g(t) (\xi_1(t)e_1 + \xi_2(t)e_2 + \xi_3(t)e_3)
$$
with boundary conditions $g(t_0)$ and $g(t_f)$,
where $g(t)=(x(t),y(t),\theta(t))$.

Finally, the regularity condition is given by the matrix
$$
A = \begin{pmatrix}
2\rho_2 & 0 & 0 & 2\rho_2 & 0 & 1 \\
0 & 2\rho_{1}\cos^{2}\gamma & 2\rho_{1}\sin\gamma\cos\gamma & 0 & -\sin\gamma & 0 \\
0 & 2\rho_{1}\sin\gamma\cos\gamma & 2\rho_{1}\sin^{2}\gamma & 0 & \cos\gamma & 1 \\
2\rho_{2} & 0 & 0 & 2\rho_{2} & 0 & 2 \\
0 & -\sin\gamma & \cos\gamma & 0 & 0 & 0 \\
1 & 0 & 1 & 2 & 0 & 0 \\
\end{pmatrix} ,
$$
whose determinant is
$$
\det A=
4\rho_{1}\rho_{2}\sin^{4}\gamma+4\rho_{1}\rho_{2}\cos^{4}\gamma+8\rho_{1}\rho_{2}\sin^{2}\gamma\cos^{2}\gamma
= 4\rho_{1}\rho_{2}(\sin^{2}\gamma+\cos^{2}\gamma)^{2}=4\rho_{1}\rho_{2}\neq 0 \, .
$$
Therefore the algorithm stabilizes at the first constraint
submanifold $\W_{c}$. Moreover, there exists a unique solution of
the dynamics, the vector field  $X\in \vf(\W_c)$ which
satisfies $\inn(X)\Omega_{\W_{c}} = \d {H}_{\W_{c}}$.
In consequence, we have a unique control input which
extremizes (minimizes) the objective function ${\mathcal A}$.
If we take the flow $F_{t} \colon \W_{c} \to \W_{c}$
of the solution vector field $X$ then we have that $F_{t}^{*}\Omega_{\W_c} = \Omega_{\W_c}$.

\section{Conclusions and further research}

We have defined, following an intrinsic point of view, the equations
of motion for constrained variational higher-order Lagrangian problems
on higher-order trivial principal bundles. As a particular case, we obtain the
higher-order Lagrange-Poincar\'e equations (see \cite{art:Colombo_Jimenez_Martin12,
art:GayBalmaz_Holm_Meier_Ratiu_Vialard12_1}). As
an interesting application we deduce the equations of motion for
optimal control of underactuated mechanical systems defined on
principal bundles. These systems appear in numerous engineering and
scientific fields. In this sense we study the optimal control of an
underactuated vehicle.

Moreover, in a future paper we will generalize the presented
construction of higher-order Euler-Lagrange equations to the case of
non-trivial principal bundles and in the context of Lie algebroids.
This last approach will be interesting because we may derive the
equations of motion for different cases as, for instance, higher-order
Euler-Poincar\'e equations, Lagrange-Poincar\'e equations and the
reduction by morphisms in a unified way
(see \cite{art:DeLeon_Marrero_Martinez05,art:Iglesias_Marrero_Martin_Sosa08,art:Saunders04}).

The case of optimal control problems for mechanical systems with
nonholonomic constraints will be also studied using some of the
ideas exposed through the paper (see \cite{art:Cortes_DeLeon_Marrero_Martin_Martinez06,
art:DeLeon_Jimenez_Martin12,art:Hussein_Bloch08} for more
details). Finally, we would like to point out that a slight
modification of the techniques presented in this work would allow to
approach the Clebsh-Pontryagin optimal control problem (see \cite{art:GayBalmaz_Holm_Ratiu13,
book:Holm08}).

\section*{Acknowledgments}

We wish to thank Prof. D. Mart\'in de Diego for fruitful
discussions and comments. We acknowledge the financial support of
the \textsl{Ministerio de Ciencia e Innovaci\'on} (Spain), projects
MTM 2010-21186-C02-01, MTM2011-22585 and  MTM2011-15725-E; AGAUR,
project 2009 SGR:1338.; IRSES-project ``Geomech-246981''; and ICMAT
Severo Ochoa project SEV-2011-0087. P.D. Prieto-Mart\'{\i}nez wants
to thank the UPC for a Ph.D grant, and L.Colombo wants to thank CSIC
for a JAE-Pre grant.


{\small
\bibliography{Bibliografia.bib}
\bibliographystyle{AMS_Mod}
}

\end{document}